\documentclass[letterpaper]{article} % DO NOT CHANGE THIS
\usepackage{aaai24}  % DO NOT CHANGE THIS
\usepackage{times}  % DO NOT CHANGE THIS
\usepackage{helvet}  % DO NOT CHANGE THIS
\usepackage{courier}  % DO NOT CHANGE THIS
\usepackage[hyphens]{url}  % DO NOT CHANGE THIS
\usepackage{graphicx} % DO NOT CHANGE THIS
\usepackage{natbib}  % DO NOT CHANGE THIS AND DO NOT ADD ANY OPTIONS TO IT
\usepackage{caption} % DO NOT CHANGE THIS AND DO NOT ADD ANY OPTIONS TO IT
\usepackage{algorithm}
\usepackage{algorithmic}

\usepackage{newfloat}
\usepackage{listings}
\DeclareCaptionStyle{ruled}{labelfont=normalfont,labelsep=colon,strut=off} % DO NOT CHANGE THIS
\floatstyle{ruled}
\newfloat{listing}{tb}{lst}{}
\floatname{listing}{Listing}
\title{Hardness of Random Reordered Encodings of Parity for Resolution and CDCL}
\author{
    Leroy Chew\textsuperscript{\rm 1},
    Alexis de Colnet\textsuperscript{\rm 1},
    Friedrich Slivovsky\textsuperscript{\rm 2},
    Stefan Szeider\textsuperscript{\rm 1}
}
\affiliations{
    \textsuperscript{\rm 1} Algorithms and Complexity Group, TU Wien, Vienna, Austria\\
    \textsuperscript{\rm 2} Department of Computer Science, University of Liverpool, Liverpool, UK\\
\{lchew,decolnet,sz\}@ac.tuwien.ac.at, f.slivovsky@liverpool.ac.uk
}

\RequirePackage{bussproofs/bussproofs}
\usepackage{todonotes}
\usepackage{url}
\usepackage{svg}
\usepackage{url}
\usepackage{enumitem} 
\usepackage{amssymb}
\usepackage{xspace}
\usepackage{pgfplots} % Load the pgfplots package
\pgfplotsset{compat=newest} % Set the pgfplots compatibility mode

\usepackage{amsthm}
\usepackage{amsmath}
\newtheorem{theorem}{Theorem}
\newtheorem{example}{Example}
\newtheorem{definition}{Definition}
\newtheorem{lemma}{Lemma}
\newtheorem{fact}{Fact}

\newcommand{\id}{\textit{id}}
\newcommand{\tw}{tw}

\newcommand{\parity}{\mathsf{Parity}}

\newcommand{\rPar}{\mathsf{rPar}}
\newcommand{\rAddPar}{\mathsf{rAddPar}}
 
\newcommand{\xor}{\mathsf{xor}}
\newcommand{\cnf}{F}

\usepackage{color,soul}
\newcommand{\SB}{\{ \,}%
\newcommand{\SM}{\mid}
\newcommand{\SE}{\,\} }%

\nocopyright

\begin{document}

\maketitle
\begin{abstract}
	Parity reasoning is challenging for Conflict-Driven Clause Learning (CDCL) SAT solvers.
	This has been observed even for simple formulas encoding two contradictory parity constraints with different variable orders (Chew and Heule 2020).
	We provide an analytical explanation for their hardness by showing that they require exponential resolution refutations with high probability when the variable order is chosen at random.
	We obtain this result by proving that these formulas, which are known to be Tseitin formulas, have Tseitin graphs of linear treewidth with high probability. Since such Tseitin formulas require exponential resolution proofs, our result follows.
	We generalize this argument to a new class of formulas that capture a basic form of parity reasoning involving a sum of two random parity constraints with random orders.
	Even when the variable order for the sum is chosen favorably, these formulas remain hard for resolution.
	In contrast, we prove that they have short DRAT refutations.
	We show experimentally that the running time of CDCL SAT solvers on both classes of formulas grows exponentially with their treewidth.
\end{abstract}

\section{Introduction}
 SAT solvers, including Conflict-Driven
Clause-Learning (CDCL) solvers, can solve practical problems with
millions of
variables~~\cite{DBLP:series/faia/SilvaLM09,FichteLeberreHecherSzeider23},
but on the other hand, can struggle with basic mathematical
principles.  The Handbook of Satisfiability \cite[Section
9.6.1]{HandbookSAT2nd} lists one such example: the problem of XOR
(exclusive-or) constraints, which is equivalent to the parity problem
of summation modulo~2.
%One such example comes from when large XOR (exclusive-or) constraints  cannot be handled efficiently by Conflict-Driven Clause-Learning (CDCL)~\cite{DBLP:series/faia/SilvaLM09}.
XOR-constraints serve practical purposes, particularly around modern
cryptographical and cryptonanalytical problems. Provably hard XOR
problems are usually constructed over complex structures such as
expander graphs~\cite{Urquhart87,BW01}, but much simpler problems
involving only two constraints were found experimentally hard for CDCL~\cite{ChewH20} and, up until this paper, were not matched with a corresponding lower bound in resolution. 

Resolution is particularly important here because the
relationship with CDCL solving is two-way; CDCL runs of unsatisfiable
instances can be output as resolution proofs, but also every
resolution refutation can be followed completely by a CDCL algorithm
(with a few non-deterministic choices) to return UNSAT~\cite{CDCLRes}. Lower bounds
on the length of resolution proof of unsatisfiability have been shown
for pure XOR problems structured by graphs and represented in CNF
formulas called Tseitin
formulas~\cite{Tseitin68,Tseitin83,Urquhart87}. Finding a complete
characterization of hard Tseitin formulas for resolution is still an
open problem. However,  exponential lower bounds for resolution are known
under the suitable condition: that their underlying graphs have high
linear \emph{treewidth} (a graph invariant that measures how close a
graph is to being a tree). The relationship between treewidth and
resolution proof length has been extensively studied for Tseitin
formulas already \cite{BW01,GalesiTT20,ItsyksonRS22,deColnetM23}.
 
Here, we look at some XOR-constraint problems that are
strikingly simple to define and whose hardness for resolution was
observed empirically but yet to be understood
theoretically~\cite{ChewH20}. We prove that they are, in fact, families of
Tseitin formulas and that linear treewidth emerges for
almost all of them, thus showing asymptotic exponential lower bounds
for resolution. Furthermore, our experiments suggest that this is not
just theoretical and asymptotic for proof systems; treewidth indeed
correlates to the solving time of CDCL solvers on these families. In the rest of this section, 
we present the problems and our results.   

\subsection{Problem 1: Reordered Parity}
The standard linear CNF encoding of an XOR-constraint over $n$ propositional variables splits the constraint into a sequence of XORs of size 3 according to an ordering of its variables~\cite[Section 2.2.5]{HandbookSAT2nd}. The encoding uses one auxiliary variable for every $k < n$ to store the parity of the first $k$ variables. The simplest form of the XOR-constraint problem starts with two opposite XOR-constraints and their standard linear CNF encodings where the $n$ variables appear in a different order given by the permutation $\sigma$. The two CNF are saying the sum of the $n$ variables is both odd and even, so their conjunction---denoted by $\rPar(n,\sigma)$---is a contradiction a SAT solver should be able to recognize.  

%The simplest form of the XOR-constraint problem starts with two identical XOR-constraints over $n$ propositional variables and two CNF encodings of the constraints where the $n$ variables appear in a different order given by the permutation $\sigma$. Suppose you flip all literals for exactly one of the variables in one of the constraints. In that case, the two encodings are now saying the sum of these variables is both odd and even, so their conjunction---denoted by $\rPar(n,\sigma)$---is a contradiction a SAT solver should be able to recognize.  
\citet{ChewH20} showed that
these problems could be proven false by $O(n\log n)$-size DRAT proofs
even without new variables~\cite{buss2019drat}. We show here that
resolution proofs on their own are often unable to handle even these
restricted examples.  For some $\sigma$, such as the identity mapping,
resolution proofs are short; in fact, the identity mapping gives the
\texttt{Dubois} family in the SAT library. However, the easiness is
not seen with other permutations.  \citet{ChewH20} conducted
experiments that showed that CDCL solvers struggle and time-out around
$n=50$ for uniformly selected permutations, although a theoretical
lower bound was never proved.  We show that as $n$ increases,
a random permutation $\sigma$ yields, with
high probability, a formula $\mathsf{rPar}(n, \sigma)$ whose
resolution proof requires exponentially many clauses.

\begin{theorem}
	There is a constant $\alpha>0$ such that, with probability tending
	to $1$ as $n$ increases, the length of a
	smallest resolution refutation of the unsatisfiable formula
	$\mathsf{rPar}(n, \sigma)$, where $\sigma$ is chosen uniformly at
	random, is at least $2^{\alpha n}$.
\end{theorem}

A key observation here is that the $\rPar$ formulas are 
Tseitin formulas. The fact that the $\rPar$ formulas come from
standard CNF encodings of very simple XOR-constraints problems makes
them natural examples of Tseitin formulas that are more likely to
occur in practice than those that appear in proofs of hardness that
are often constructed from arbitrary expander
graphs~\cite{Urquhart87}. Theorem~\ref{theorem:hard_formulas} proves
the context assumed by \citet{ChewH20} that a powerful proof system
such as DRAT$^-$ was necessary for the short proofs of $\rPar$, as we
now have exponential resolution lower bounds. Along with the evidence
shown by our experiments, we now conclude that high treewidth is the
reason for the hardness in the experiments of \citet{ChewH20}. We can
also take this as clarification that order matters in the encoding of
parity constraints in general.

\subsection{Problem 2: Random Parity Addition}

There are several effective strategies for dealing with
XOR-constraints in practice.  One method that has succeeded is to
employ Gaussian elimination~\cite{HJ12,Soos12} techniques to simplify
the problem.  Two contradictory parity constraints fall short of
representing what happens in an average step of Gaussian
elimination. Instead, Gaussian elimination involves many steps using
the bitwise addition of two XOR-constraints to produce a new
constraint.  In order to study such a step as an instance for
resolution, we have to write it as a contradiction. So here we modify
$\rPar$ to use three XOR-constraints, with the third containing 
the variables in the symmetric difference of the first two
constraints and then flip some literals to create a
contradiction. The input XOR-constraints are encoded in
CNF formulas $a$ and $b$ using the standard linear
encoding. We then define a CNF encoding
$\rAddPar(a,b)$ of the contradiction similar to
$\mathsf{rPar}(n, \sigma)$, and we show their hardness for resolution.

\begin{theorem}
	With high probability, for any two random parity constraints over $n$ variables encoded randomly and independently in CNF formulas $a$ and $b$ using the standard linear encoding, the length of the shortest resolution refutation of $\mathsf{rAddPar}(a,b)$ is exponential in $n$, the length of a shortest resolution refutation of $\mathsf{rAddPar}(a,b)$ is exponential in the number of variables.
\end{theorem}

The $\rAddPar$ formulas turn out to also be Tseitin formulas, so this again provides a new intuitive family that demonstrates the hardness of Tseitin formulas---and yet again shows that order matters when encoding parity constraints.

Adding Gaussian elimination to SAT-solving/\hspace{0pt}preprocessing
presents several technical challenges.  An example is
verification---unsatisfiable instances in CDCL SAT solvers
can be readily verified in resolution proofs and thus verified in the
more powerful checking format standard DRAT~\cite{JHB12}. It was
therefore pertinent to show that Gaussian elimination techniques could
also be verified efficiently in DRAT \cite{PR16}. For the specific
family of $\mathsf{rPar}(n, \sigma)$, it was shown to have DRAT$^-$
refutations in $O(n\log n)$ many lines using a tool from
\citet{ChewH20}. Recently, a BDD-based SAT solver augmented with
pseudo-Boolean constraints \cite{BBH22} was shown to have improved the
result experimentally.  We can generalize Chew and Heule's upper-bound
results to $\rAddPar(a,b)$.

\begin{theorem}
	For any two random parity constraints over $n$ variables
        encoded randomly and independently in CNF formulas $a$ and $b$
        using the standard linear encoding, there are DRAT$^-$
        refutations of $\rAddPar(a,b)$ with  $O(n\log n)$ many lines.
\end{theorem}

$O(n\log n)$ is already a good upper bound, and this can potentially
be used in verification.  One of the advantages of DRAT$^-$ is that no
extension variables are added. This will speed up the proof checking.

\iftrue
\begin{theorem}
	For any two parity constraints over $n$ variables, and any standard linear encodings $a$ and $b$ of theses constraints, there are Extended Resolution refutations of $\rAddPar(a,b)$ that have $O(n\log n)$ many lines.
\end{theorem}
\fi

\section{Preliminaries}
Boolean \emph{variables} take value in $\{0,1\}$. A \emph{literal} is either a variable $x$ or its negation $\bar x$. \emph{Clauses} are disjunctions of literals and \emph{CNF formulas} are conjunctions of clauses.
The negation of clause $C$ can be labelled $\bar C$ and is a CNF of clauses each containing one literal. The symbols $\lor$, $\land$ denote disjunction and conjunction and we use~$\oplus$ for exclusive disjunction, that is, $x \oplus y = x + y \text{ mod } 2$. The \emph{canonical CNF representation} of a \emph{parity constraint} $x_1 \oplus \dots \oplus x_k = 0$ is the CNF formula $\xor(x_1,\dots,x_k)$ composed of all $2^{k-1}$ clauses of size $k$ that contain an odd (resp. even) number of positive literals when $k$ is odd (resp. even). For instance
\[\xor(p,q,r):= (\bar p \vee \bar q \vee \bar r) \wedge (\bar p \vee  q \vee  r) \wedge ( p \vee \bar q\vee  r) \wedge ( p \vee q \vee \bar r). \]
The canonical representation of $x_1 \oplus \dots \oplus x_k = 1$ is just $\xor(x_1,\dots,x_k)$ where we flip all literals for an arbitrary variable, for instance $\xor(\bar x_1,x_2,\dots,x_k)$.
\subsection{Proofs and Refutations}

\subsubsection{Resolution.}

\emph{Resolution} is a refutational proof system that works by adding clauses based on a single binary rule---the resolution rule~\cite{Rob63}.
The resolution rule's new clause is a logical implication. Adding it to the formula preserves not only satisfiability but also the  models.
A resolution proof that derives the empty clause shows that the original formula is unsatisfiable. 

\begin{prooftree}
	\AxiomC{$C_1 \vee x$}
	\AxiomC{$C_2 \vee \neg x$}
	\RightLabel{(Resolution)}
	\BinaryInfC{$C_1\vee C_2$}
\end{prooftree}

\subsubsection{Extended Resolution.}

\emph{Extended resolution} adds an extension rule, it creates extension clauses that introduce a new variable with clauses that force that new variables to follow a definition. If we treat the extension variables as new, the extension rule does not change the satisfying assignments when considering only the original variables. Therefore, when we reach the empty clause, we know that the original formula must have been unsatisfiable. 

\begin{example}
	We can add the following extension clauses that state that extension variable $n$ is the exclusive or of $x$ with~$y$:
	$(\bar x \vee y \vee n),	(x \vee \bar y \vee n),
	(x \vee y \vee \bar n),	(\bar x \vee \bar y \vee \bar n)$.

\end{example}

\subsubsection{DRAT.}
Unit propagation is an incomplete, model preserving, and  polynomial-time process.
\begin{definition}
	A unit clause is a clause of one literal. 
	\emph{Unit propagation} takes any unit clause  $(a)$ and resolves it with every clause which has an $\bar a$ in it (possibly creating another unit clause). After all resolvents are found, the clause $(a)$ is removed, and we repeat the process for another unit clause until no unit clauses remain. We also terminate if we reach the empty clause, and we can write $\cnf \vdash_1 \bot$ to denote that unit propagation of $\cnf$ reaches the empty clause. 
\end{definition}

While unit propagation itself is incomplete, it terminates in
polynomial time. It therefore is a convenient tool for
checking implication, we can use this in the concept of an asymmetric
tautology, which is a clause that must be true assuming a CNF  because its negation would cause a conflict via unit propagation. 
\begin{definition}[\citealt{JHB12}]\label{def:ATA}
	Let $\cnf$ be a CNF formula. A clause $C$ is an \emph{asymmetric tautology} (AT) w.r.t.~$F$ if $\cnf\wedge  \bar{C}\vdash_{1} \bot $.
\end{definition}

A clause being an asymmetric tautology in $\cnf$ is a generalization of being
a resolvent of some pair of clauses in~$\cnf$.
We also want to be able to generalize the creation of extension clauses. 
To do this, we first generalize extension clauses to blocked clauses. Blocked clauses are clauses that have a literal that cannot be resolved without causing a tautology, and so they are non-threatening to the satisfiability of the formula. 
We generalize blocked clauses to RAT clauses, where we widen the condition of tautology to asymmetric tautology.  

\begin{definition}[\citealt{JHB12}]
	Let $\cnf$ be a CNF formula. A clause $C$ is a \emph{resolution asymmetry tautology} (RAT) w.r.t.~$F$ if
	there exists a literal $l \in C$ such that 
	for every clause $\bar{l}\vee D \in \cnf$ it holds that
	$\cnf\wedge \bar{D} \wedge \bar{C}\vdash_{1} \bot $.
\end{definition}

DRAT is a generalized and application friendly version of extended resolution. Each rule modifies a formula  by either adding (removing) a clause while preserving satisfiability (unsatisfiability), respectively. Unlike resolution, clauses can be added (removed) with or without preserving the exact set of satisfying models of a formula. 
The first set of DRAT rules show us how we can add or remove clauses while preserving models by using asymmetric tautologies when  $C$ is AT w.r.t.~$\cnf$:
	\begin{minipage}{.24\textwidth}
		\begin{prooftree}
			\AxiomC{$\cnf$}
			\RightLabel{(ATA)}
			\UnaryInfC{$\cnf \wedge C$}
		\end{prooftree}	
	\end{minipage}\begin{minipage}{.24\textwidth}
		\begin{prooftree}
			\AxiomC{$\cnf  \wedge C$}
			\RightLabel{(ATE)}
			\UnaryInfC{$\cnf$}
		\end{prooftree}
	\end{minipage}
\smallskip 

\noindent The second set of rules use resolution asymmetric tautologies  ($C$ is RAT w.r.t.~$\cnf$) and do not preserve models:
	\begin{minipage}{.24\textwidth}
		\begin{prooftree}
			\AxiomC{$\cnf$}
			\RightLabel{(RATA)}
			\UnaryInfC{$\cnf \wedge C$}
		\end{prooftree}	
	\end{minipage}\begin{minipage}{.24\textwidth}
		\begin{prooftree}
			\AxiomC{$\cnf  \wedge C$}
			\RightLabel{(RATE)}
			\UnaryInfC{$\cnf$}
		\end{prooftree}	
	\end{minipage}
\smallskip 
	
\noindent In all clause additions, we can add a new variable as long as it works with the side conditions of the rules. 
However, an excess of new variables can cause a proof checker to slow down, so there is a version of DRAT that forbids new variables known as DRAT$^-$.

\iftrue
\subsection{Treewidth}

A tree decomposition $(T,\lambda)$ of $G$ consists of a tree $T$ equipped with a mapping $\lambda : V(T) \rightarrow \mathcal{P}(V(G))$ such that
\begin{itemize}
	\item $\bigcup_{t \in V(T)} \lambda(t) = V(G)$,
	\item for every $uv \in E(G)$, there is $t \in V(T)$ such that $\{u,v\} \subseteq \lambda(t)$,
	\item for every $v \in V(G)$, $T[\SB t \SM v \in \lambda(t)\SE]$ is a connected subtree of $T$.
\end{itemize}
The \emph{width} of $(T,\lambda)$ is $\max_{t \in V(T)} |\lambda(t)|-1$. The treewidth of $G$, denoted by $\tw(G)$, is the smallest width over all tree decompositions of $G$.
\fi

\subsection{Tseitin Formulas}

A \emph{Tseitin formula} is a CNF formula that represents a system of parity constraints where every variable appears in exactly two constraints. Such a formula is determined by a graph~$G$: each edge $e$ corresponds to a unique Boolean variable $x_e$ and each vertex $v$ defines a constraint $\bigoplus_{e \in E(v)} x_e = c(v)$, where $E(v)$ is the set of edges incident to $v$ in $G$ and $c : V(G) \rightarrow \{0,1\}$ is the \emph{charge function}. The Tseitin formula $T(G,c)$ is the conjunction of the $\xor$ representations for the constraints for every $v \in V(G)$. We call $G$ the \emph{Tseitin graph} of the formula. It is often assumed that the maximum degree of all vertices in $G$ is bounded by a constant, so that the size of $T(G,c)$ is linear in $|var(T(G,c))| = |E(G)|$.

\begin{example}\label{example:tseitin_formula}
	Let $G$ be the following graph with $V(G) = \{1,2,3,4\}$. Let $c : V(G) \rightarrow \{0,1\}$ such that gray vertices have charge $0$ and white vertices have charge $1$.
\begin{minipage}{0.25\textwidth}
\begin{itemize}
\item[] $x_{12} \oplus x_{13} \oplus x_{14} = 0$
\item[] $x_{12} \oplus x_{23} = 1$
\item[] $x_{13} \oplus x_{23} \oplus x_{34} = 1$
\item[] $x_{14} \oplus x_{34} = 0$
\end{itemize}
\end{minipage}\begin{minipage}{0.25\textwidth}
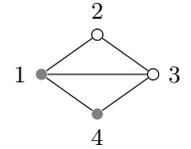
		\begin{center}
			\begin{tikzpicture}[xscale=0.9, yscale=0.8]
				\def\s{1.5};
				\node[circle, fill=gray, inner sep=\s, label=left:{\small $1$}] (a) at (0,0) {}; 
				\node[circle, draw=black, inner sep=\s, label=above:{\small $2$}] (b) at (0.825,0.66) {}; 
				\node[circle, draw=black, inner sep=\s, label=right:{\small $3$}] (c) at (1.65,0) {}; 
				\node[circle, fill=gray, inner sep=\s, label=below:{\small $4$}] (d) at (0.825,-0.66) {}; 
				
				\draw (a) -- (b) -- (c) -- (d) -- (a) -- (c);
			\end{tikzpicture}
		\end{center}
	\end{minipage}
	\noindent The Tseitin formula for this graph and this charge $c$ is 
	\begin{multline*}
	T(G,c) = \xor(x_{12},x_{13},x_{14}) \land \xor(\bar x_{12},x_{23}) \\ \land
		\xor(\bar x_{13},x_{23},x_{34}) \land \xor(x_{14},x_{34}).
	\end{multline*}
\end{example}

Tseitin formulas were introduced by  \citet{Tseitin68,Tseitin83} in the 1960s as hard instances for proof systems, despite an easy criterion to decide their satisfiability~\cite[Lemma 4.1]{Urquhart87}. 
\citet{Urquhart87}  later showed that when $G$ belongs to the family of bounded-degree expander graphs (whose definition we omit), all resolution refutations of $T(G,c)$ require exponentially many clauses. This was generalized by Ben-Sasson and Widgerson who used their width-length relations on refutation proofs to derive exponential lower bounds parameterized on the edge expansion of $G$~\cite{BW01}. Beyond expansion, the key parameter to characterize the hardness of Tseitin formulas for resolution could be the treewidth of the graph. 
Treewidth is a very well-known graph parameter whose definition we omit (see~\cite{Bodlaender98}). Intuitively the treewidth of $G$, denoted by $tw(G)$, is an integer between $0$ and $|V(G)|$ that measures how close $G$ is to a tree (trees having treewidth~$1$). 
On the one hand, it was shown~\cite{AlekhnovichR11}  that
unsatisfiable Tseitin formulas have resolution refutations of length
at most $2^{O(\tw(G))}|E(G)|^{O(1)}$, thus a logarithmic
treewidth guarantees short refutations. On the other hand, combining
the width-length relation with Corollaries 8 and 16 of Galesi et al.'s \shortcite{GalesiTT20} yields the following:

\begin{theorem}\label{theorem:linear_tw_res}
	Let $G$ be an $n$-vertex graph whose maximum degree is bounded by a constant, if
	$\tw(G) = \Omega(n)$, then the length of a shortest resolution
	refutation of an unsatisfiable Tseitin formula $T(G,c)$ is at least
	$2^{\Omega(n)}$.
\end{theorem}
\begin{proof}
Using~\cite[Corollary 3.6]{BW01}, the length of the shortest resolution refutation of any unsatisfiable CNF formula $F$ is at least $\exp(\Omega((w(F,\vdash) - w(F))^2/n))$ where $w(F)$ is the length of the longest clause in $F$, and $w(F,\vdash)$ is the largest $k$ such that every refutation of $F$ by resolution contains a clause of length $k$. Now let $\Delta = O(1)$ be the maximum degree of $G$. Then $w(T(G,c)) = \Delta$ and the length of the shortest resolution refutation of $T(G,c)$ is in $\exp(\Omega(w(T(G,c),\vdash)^2/n))$.

By~\cite[Corollaries 8 and 16]{GalesiTT20}, $w(T(G,c),\vdash) = \max(\Delta,\tw(L(G)))$ where $L(G)$ is the \emph{line graph} of $G$ (see~\cite[Definition 6]{GalesiTT20} if needed). It is well-known that $\tw(L(G)) \geq \frac{1}{2}(\tw(G)+1)-1$, see for instance~\cite{HarveyW18}. So $\tw(G) = \Omega(n)$ implies $\tw(L(G)) = \Omega(n)$, which implies $w(T(G,c),\vdash) = \Omega(n)$. And the theorem follows.
\end{proof}

\noindent Note that there is still a gray area for $\tw(G)$ less than linear, but more than logarithmic in $n$. Note also that Tseitin formulas are easily refutable in proof systems different from resolution, regardless of high treewidth~\cite{ItsyksonKRS20,BonacinaBL23}.

\section{Parity Problems}

For a constraint $x_1 \oplus \dots \oplus x_n = c$, the number of clauses in the canonical representation is exponential in $n$. 
We can use the $\xor$ notation to build larger parity constraints if we include auxiliary variables called Tseitin variables:

\begin{definition}
	Let $\sigma$ be a permutation of $n$ elements and $X= \SB x_i \SM 1
	\leq i \leq n\SE$ be an ordered set of literals. We define
	$\parity(X, T, \sigma)= \xor(x_{\sigma(1)}, x_{\sigma(2)}, t_1)\land
	{\bigwedge_{j=1}^{n-4}  \xor(t_{j}, {x_{\sigma(j+2)}},t_{j+1})}\land \xor({t}_{n-3},x_{\sigma(n-1)}, x_{\sigma(n)})$, 
	where $T= \SB t_i \SM i\leq n-3\SE$ are Tseitin variables. 
\end{definition}

Here $\parity(X, T, \sigma)$ is satisfiable if and only if the total parity of $X$ is $0$. If we wanted a constraint which is satisfiable if and only if the parity was $1$, again we simply flip a literal.
In our particular Tseitin encoding, we structure the~$\oplus$ linearly so that the formula for $n=5, \sigma=\id$ looks like $((((x_1 \oplus x_2)\oplus x_3)\oplus x_4) \oplus x_5)$. Here the formula depth is linear, however the structure does not affect the satisfiability as $\oplus$ is associative. Furthermore the actual permutation $\sigma$ does not affect the satisfiability because $\oplus$ is commutative. 

\subsection{Problem 1: Reordered Parity}
In our first problem we simply take two parity constraints that are in contradiction. 
We simultaneously state that the variables of $X$ have $0$ parity and $1$ parity.
This is obviously a contradiction.
However in order to make it difficult we use two different permutations to obscure the conflict. We define
\[
\rPar(n, \sigma) = \parity(X, S, \id) \wedge \parity(X', T, \sigma)
\] 
where $X=\SB x_i \SM 1 \leq i \leq n\SE$, $X'=\SB x_i \SM 1 \leq i <
n\SE \cup \{\bar x_n \}$, $\sigma$ is a permutation of $n$ elements, and $\id$ is the identity map. $S$, $T$, and $X$ are disjoint sets of variables. Note that $\rPar(n, \sigma)$ is a Tseitin formula where each $\xor$ constraint corresponds to a vertex of the underlying graph. This is because every variable appears in exactly two $\xor$ constraints. For $n$ fixed, its Tseitin graph depends only on $\sigma$ and its vertices all have charge $0$, except for the vertex corresponding to $\xor(t_{n-3},x_{\sigma(n-1)},\bar x_{\sigma_n})$.

\begin{fact}
$\rPar(n, \sigma)$ is an unsatisfiable Tseitin formula.
\end{fact}

The version where the identity map is used for $\sigma$ is the most natural, and is easy to solve. However the random version can still appear from equivalences (which themselves are xors) obfuscating the random parity. For example: A system of equations containing
$x_1 \oplus x_2 \oplus x_3 \oplus x_4 = 0$ and
$x_5 \oplus x_6 \oplus x_7 \oplus x_8 = 1$
and other binary clauses implying: $x_5 \leftrightarrow x_4$ and $x_6 \leftrightarrow x_1$ and $x_7 \leftrightarrow x_2$ and $x_8 \leftrightarrow x_3$.
A standard CNF encoding using the natural variable ordering $x_1 < x_2 < x_3 < x_4 < x_5 < x_6 < x_7 < x_8$ (as done in the Dubois benchmark family) yields a non-trivial random parity problem after removing binary clauses, namely:
$x_1 \oplus x_2 \oplus x_3 \oplus x_4 = 0$ encoded in CNF using the ordering $x_1 < x_2 < x_3 < x_4$, and
$x_4 \oplus x_1 \oplus x_2 \oplus x_3 = 1$ encoded in CNF using the ordering $x_4 < x_1 < x_2 < x_3$.

\subsection{Problem 2: Random Parity Addition}
Problem 1 is a simple special case. 
In general, solvers and preprocessors want to deal with XOR-constraints by Gaussian elimination. In Gaussian elimination we use multiple steps involving adding two parity constraints together to get a third parity constraint. 
Since the parity constraints may have a large number of input variables we would have to use Tseitin variables, including in the sum constraint. 
We model the difficulty of an addition step by taking the three parity constraints: the two summands and the negation of the sum, in conjunction to get a contradiction.

Let $a= \parity(A,S, \sigma_a )$ and $b= \parity(B,T, \sigma_b )$,
where $A$ and $B$ are subsets of $X=\SB x_i \SM 1\leq i \leq n\SE$ and $S$ and $T$ are disjoint sets of Tseitin variables. We define 
\begin{multline*}
\rAddPar(a,b,\sigma_c) = \parity(A,S, \sigma_a ) \\ \land \parity(B,T, \sigma_b ) \land \parity (C,U,\sigma_c)
\end{multline*}
where $U$ is disjoint from $S$, $T$ and $X$. Here $C$ is the symmetric difference of $A$ and $B$ but with a literal flipped. In a first scenario, the variable ordering $\sigma_c$ for the sum constraint is independent of that used in $a$ and $b$. This is modeled fixing $\sigma_c$ to be the identity $\id$. For convenience we write $\rAddPar(a,b) = \rAddPar(a,b,\id)$. A more clever approach is to choose $\sigma_c$ favorably for $\sigma_a$ and $\sigma_b$. We will precise what we mean by ``choosing $\sigma_c$ favorably'' later in the paper.

Again $\rAddPar(a,b,\sigma_c)$ is a Tseitin formula since every variable appears in two $\xor$ constraints. Every variable appearing in $A$ and $B$ does not appear in the third constraint, and every variable in the symmetric difference of $A$ and $B$ appears a second time the third constraint. The Tseitin variables are disjoint so they also appear in exactly two $\xor$.

\begin{fact}
$\rAddPar(a,b,\sigma_c)$ is an unsatisfiable Tseitin formula.
\end{fact}

\section{The Graph Model and Lower Bounds}\label{sec:lower}

In this section we present our lower bounds on the length of resolution refutations for $\rPar$ and $\rAddPar$ when they are constructed in a random fashion.

\subsection{Lower Bounds for Reordered Parity}

In the following, we denote the set $\{1,\dots,n\}$ by $[n]$ and we call $\mathfrak{S}_n$ the set of permutations of $[n]$.
Consider $2n$ vertices labeled $1,\dots,n,1',\dots,n'$, a permutation
$\sigma \in \mathfrak{S}_n$, and let $G_\sigma$ be the graph over
these vertices whose edge set is $\SB (i,i+1) \SM i < n\SE \cup \SB
(\sigma(i)',\sigma(i+1)') \SM i < n\SE \cup \SB (i,i') \SM i \leq n\SE$. Let $G^*_\sigma$ be the multigraph obtained by contracting the edges $(i,i')$ for all $i \in [n]$. That is, $V(G^*_\sigma) = [n]$ and, for every edge $(i,j)$, $(i',j)$, $(i,j')$ or $(i',j')$ in $E(G)$, we add an edge $(i,j)$ to $E(G^*_\sigma)$.

\begin{example}\label{example:G_sigma}
	Let $n = 5$ and  $\sigma(1) = 3$, $\sigma(2) = 1$, $\sigma(3) = 5$, $\sigma(4) = 4$, $\sigma(5) = 2$. The graph $G_\sigma$ and $G^*_\sigma$ are:

			$$G_\sigma = \raisebox{-0.5\height}
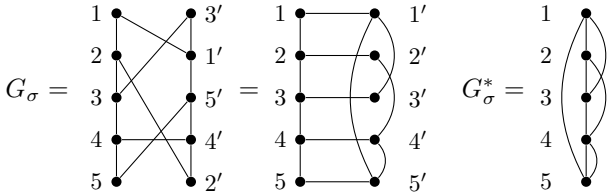
{
			\begin{tikzpicture}[xscale=0.9, yscale=0.8, every node/.style={scale=0.9}]
				\def\x{1.1};
				\def\y{0.7};
				\def\s{1.5};

				\node[circle,fill=black, inner sep=\s, label=left:{$1$}] (1) at (0,0) { };
				\node[circle,fill=black, inner sep=\s, label=left:{$2$}] (2) at (0,-1*\y) { };
				\node[circle,fill=black, inner sep=\s, label=left:{$3$}] (3) at (0,-2*\y) { };
				\node[circle,fill=black, inner sep=\s, label=left:{$4$}] (4) at (0,-3*\y) { };
				\node[circle,fill=black, inner sep=\s, label=left:{$5$}] (5) at (0,-4*\y) { };  
				
				\node[circle,fill=black, inner sep=\s, label=right:{$3'$}] (1') at (\x,0) { };
				\node[circle,fill=black, inner sep=\s, label=right:{$1'$}] (2') at (\x,-1*\y) { };
				\node[circle,fill=black, inner sep=\s, label=right:{$5'$}] (3') at (\x,-2*\y) { };
				\node[circle,fill=black, inner sep=\s, label=right:{$4'$}] (4') at (\x,-3*\y) { };
				\node[circle,fill=black, inner sep=\s, label=right:{$2'$}] (5') at (\x,-4*\y) { };  
				
				\draw (1) -- (2) -- (3) -- (4) --(5);
				\draw (1') -- (2') -- (3') -- (4') --(5');
				
				\draw (1) -- (2');
				\draw (2) -- (5');
				\draw (3) -- (1');
				\draw (4) -- (4');
				\draw (5) -- (3');
			\end{tikzpicture}} = 
			\raisebox{-0.5\height}{\begin{tikzpicture}[xscale=0.9, yscale=0.8, every node/.style={scale=0.9}]
				\def\x{1.1};
				\def\y{0.7};
				\def\s{1.5};

				\node[circle,fill=black, inner sep=\s, label=left:{$1$}] (1) at (0,0) { };
				\node[circle,fill=black, inner sep=\s, label=left:{$2$}] (2) at (0,-1*\y) { };
				\node[circle,fill=black, inner sep=\s, label=left:{$3$}] (3) at (0,-2*\y) { };
				\node[circle,fill=black, inner sep=\s, label=left:{$4$}] (4) at (0,-3*\y) { };
				\node[circle,fill=black, inner sep=\s, label=left:{$5$}] (5) at (0,-4*\y) { };  
				
				\node[circle,fill=black, inner sep=\s, label={[label distance=0.3cm]right:{$3'$}}] (1') at (\x,-2*\y) { };
				\node[circle,fill=black, inner sep=\s, label={[label distance=0.3cm]right:{$1'$}}] (2') at (\x,0*\y) { };
				\node[circle,fill=black, inner sep=\s, label={[label distance=0.3cm]right:{$5'$}}] (3') at (\x,-4*\y) { };
				\node[circle,fill=black, inner sep=\s, label={[label distance=0.3cm]right:{$4'$}}] (4') at (\x,-3*\y) { };
				\node[circle,fill=black, inner sep=\s, label={[label distance=0.3cm]right:{$2'$}}] (5') at (\x,-1*\y) { };  
				
				\draw (1) -- (2) -- (3) -- (4) --(5);
				\draw (1') to [out=50, in=-50] (2');
				\draw (2') to [out=-115, in=115] (3');
				\draw (3') to [out=50, in=-50] (4');
				\draw (4') to [out=50, in=-50] (5');
				
				\draw (1) -- (2');
				\draw (2) -- (5');
				\draw (3) -- (1');
				\draw (4) -- (4');
				\draw (5) -- (3');
			\end{tikzpicture}}
			\quad G^*_\sigma = 
			\raisebox{-0.5\height}{\begin{tikzpicture}[xscale=0.9, yscale=0.8, every node/.style={scale=0.9}]
				\def\x{1};
				\def\y{0.7};
				\def\s{1.5};
				\node[circle,fill=black, inner sep=\s, label={[label distance=0.3cm]left:{$1$}}] (1a) at (10,0) { };
				\node[circle,fill=black, inner sep=\s, label={[label distance=0.3cm]left:{$2$}}] (2a) at (10,-1*\y) { };
				\node[circle,fill=black, inner sep=\s, label={[label distance=0.3cm]left:{$3$}}] (3a) at (10,-2*\y) { };
				\node[circle,fill=black, inner sep=\s, label={[label distance=0.3cm]left:{$4$}}] (4a) at (10,-3*\y) { };
				\node[circle,fill=black, inner sep=\s, label={[label distance=0.3cm]left:{$5$}}] (5a) at (10,-4*\y) { };

				\draw (1a) -- (2a) -- (3a) -- (4a) --(5a);
				\draw (3a) to [out=50, in=-50] (1a);
				\draw (1a) to [out=-115, in=115] (5a);
				\draw (5a) to [out=50, in=-50] (4a);
				\draw (4a) to [out=50, in=-50] (2a);
			\end{tikzpicture}}$$
\end{example} 
\medskip
The maximum degree of a vertex of $G_\sigma$ (resp. $G^*_\sigma$) is $3$ (resp. $4$). Since $G^*_\sigma$ is a minor of $G_\sigma$ after merging of the parallel edges, we have that $\tw(G^*_\sigma) \leq \tw(G_\sigma)$. We also have a bound in the other direction which may be useful when $\tw(G_\sigma)$ is harder to compute than $\tw(G^*_\sigma)$ in practice.

\begin{lemma}\label{lemma:tw_after_contraction}
	We have that $\frac{1}{2}\tw(G_\sigma) \leq \tw(G^*_\sigma) \leq \tw(G_\sigma)$. 
\end{lemma}
\iftrue
\begin{proof}
	We only have to show that $\tw(G_\sigma) \leq
	2\cdot \tw(G^*_\sigma)$. Let $(T,\lambda^*)$ be a tree decomposition
	of $G^*_\sigma$ and let $\lambda$ be defined by, $\forall t \in T$, $\lambda(t) = \SB i \SM i \in
	\lambda^*(t) \SE \cup \SB i' \SM i \in \lambda^*(t)\SE$. We
        show that $(T,\lambda)$ is a tree decomposition of
        $G_\sigma$. First it is clear that $\bigcup_{t \in T}
        \lambda(t) = V(G_\sigma)$. Second, for every $(i,j)$,
        $(i,j')$, $(i',j)$ or $(i',j')$ in $G_\sigma$ there is a $t
        \in T$ such that $\lambda^*(t)$ contains $(i,j)$ so the four
        aforementioned edges belong to $\lambda(t)$. Finally, for
        every $i \in [n]$ we have that $T[\SB t \SM i \in
        \lambda(t)\SE] = T[\SB t \SM i' \in \lambda(t)\SE] = T[\SB t
        \SM i \in \lambda^*(t)\SE]$ and we know that $T[\SB t \SM i \in \lambda^*(t)\SE]$ is connected. The width of $(T,\lambda)$ is twice that of $(T,\lambda^*)$, so $\tw(G_\sigma) \leq 2\cdot \tw(G^*_\sigma)$.
\end{proof}
\fi
The Tseitin graph of $\rPar(n, \sigma)$ is not exactly $G_\sigma$. The two graphs would be the same if we were to slightly modify the first and last constraints of $\parity(X,S,\id)$ and $\parity(X',T,\sigma)$ by replacing $\xor(x_1,x_2,s_1)$ by $\xor(x_1,\bar s_0) \land \xor(s_0,x_2,s_1)$ etc.

\begin{lemma}
The Tseitin graph of $\rPar(n,\sigma)$ is obtained by contracting four edges of $G_\sigma$.
\end{lemma}
\begin{proof}[Proof sketch]
Contract the  edges $(1,2)$, $(n-1,n)$, $(\sigma(1)',\sigma(2)')$ and $(\sigma(n-1)',\sigma(n)')$ of $G_\sigma$. 
\end{proof}

Since an edge contraction can only decrease the treewidth by one, it follows that the treewidth of the Tseitin graph of $\rPar(n,\sigma)$ is at least $\tw(G_\sigma) - 4$. But then we show that when $\sigma$ is sampled uniformly at random from $\mathfrak{S}_n$, with high probability (i.e., with probability tending to $1$ as $n$ increases) both $G_\sigma$ and $\rPar(n,\sigma)$ have linear treewidth.

\begin{lemma}\label{lemma:expander_whp}
	There is a constant $\alpha > 0$ such that 
	%$\textup{Pr}(i(G_\sigma) < \alpha)$ and 
	$\textup{Pr}(\tw(G_\sigma) < \alpha n)$ vanishes to $0$ as $n$ increases when $\sigma$ is chosen uniformly at random in $\mathfrak{S}_n$.
\end{lemma}
\begin{proof}
  \citet{KimW01} have studied the graph distribution
  $\mathcal{H}_n \oplus \mathcal{H}_n$ where each graph over $n$
  vertices is the superposition of two independent hamiltonian cycles
  over these vertices (merging parallel edges). They call
  $\mathcal{G}_{4,n}$ the uniform distribution over all 4-regular
  graphs over $n$ vertices. They show that any sequence of events is
  true asymptotically almost surely (a.a.s.)  in
  $\mathcal{H}_n \oplus \mathcal{H}_n$ if and only if it is true
  a.a.s. in $\mathcal{G}_{4,n}$ \cite[Theorem 2]{KimW01}.
	
	The treewidth of a random 4-regular graph from $\mathcal{G}_{4,n}$ is linear in $n$ with high probability~\cite{ChandranS03}. On the one hand, Chandran and Subramanian (\citeyear{ChandranS03}) have shown that the treewidth of a $d$-regular graph $G$ is at least $\lfloor \frac{3n}{4}\frac{(d-\lambda_2(G))}{(3d-2\lambda_2(G))}\rfloor -1$ where $\lambda_2(G)$ is the second largest eigenvalue of the adjacency matrix of~$G$. On the other hand, Friedman has shown that, for any fixed $\varepsilon > 0$ and any $d \geq 2$, $|\lambda_2(G)| \leq 2\sqrt{d-1} + \varepsilon$ holds with high probability when $G \in \mathcal{G}_{d,n}$~\cite[Corollary 1.4]{Friedman03}\footnote{Note that our $\mathcal{G}_{d,n}$ is Friedman's $\mathcal{K}_{d,n}$}. The combination of the two results yields that, with high probability when $G \in \mathcal{G}_{4,n}$, $\tw(G) \geq \Omega(n)$.
	
	Thus, with high probability, a random graph $G \in \mathcal{H}_n \oplus \mathcal{H}_n$ has linear treewidth. Now $G^*_\sigma$ is not the superposition of two independent hamiltonian cycles but the superposition of two independent paths. But if we close both paths before superposition, then we obtain a graph in $\mathcal{H}_n \oplus \mathcal{H}_n$ whose treewidth is at least $\tw(G^*_\sigma) - 2$. So with high probability $\tw(G_\sigma) \geq \tw(G^*_\sigma) \geq \Omega(n)$ holds.
\end{proof}

$\rPar(n,\sigma)$'s graph has degree at most $4$ and linear treewidth with high probability so, by Theorem~\ref{theorem:linear_tw_res} we immediately have that, with high probability, the Tseitin formula is hard for resolution.

\setcounter{theorem}{0}

\begin{theorem}\label{theorem:hard_formulas}
	There is a constant $\alpha > 0$ such that, with probability tending to $1$ as $n$ increases, the length of a shortest resolution refutation of $\rPar(n, \sigma)$ where $\sigma$ is chosen uniformly at random in $\mathfrak{S}_n$, is least $2^{\alpha n}$.  
\end{theorem}

\subsection{Lower Bounds for Random Parity Addition}

The symmetric difference of two subsets $A$ and $B$ of $X =
\{x_1,\dots,x_n\}$ is denoted by $A \triangle B := (A \cup B)
\setminus (A \cap B)$. Recall that for $a = \parity(A,S, \sigma_a )$
and $b= \parity(B,T, \sigma_b )$, the formula $\rAddPar(a,b) = a \land
b \land \parity(C,U,\id)$ is a Tseitin formula (where $U \cap S = U
\cap T = \emptyset$ and $C$ is $A \triangle B$ with one literal
flipped). Let us describe $\rAddPar(a,b)$'s Tseitin graph. We call $H$ the graph whose vertices are split into three sets $V = \SB i \SM x_i \in A\SE$, $V'
= \SB i' \SM x_i \in B\SE$ and $V'' = \SB i'' \SM x_i \in A \triangle B\SE$. The edge set of $H$ contains $(i,i')$ for all $x_i \in A \cap B$, and  $(i,i'')$ for all $x_i \in A \cap (A\triangle B)$, and $(i',i'')$ for all $x_i \in B \cap (A\triangle B)$. The vertices in $V$ (resp. $V'$ and $V''$) are also connected in a path following the order $\sigma_a$ (resp. $\sigma_b$ and $\id$). $H$ is not exactly the Tseitin graph of $\rAddPar$, but is close enough and easier to analyze.

\begin{lemma}\label{lemma:tw_after_contraction_2}
For $a = \parity(A,S, \sigma_a )$ and $b= \parity(B,T,$ $\sigma_b )$ the Tseitin graph of $\rAddPar(a,b)$ is obtained by contracting six edges of the graph $H$.
\end{lemma}

\noindent At this point it is worth giving an example of such a graph~$H$.
\begin{example}\label{example:3_layer_graph}
	Let $n = 6$, $A = \{x_1,x_2,x_4,x_5,x_6\}$ and $B = \{x_1,x_2,x_3,x_5\}$. So $A \triangle B = \{x_3,x_4,x_6\}$. The constraints encoded in CNF with $\parity$ are $x_1 \oplus x_2 \oplus x_4 \oplus x_5 \oplus x_6 = 0$, $x_1 \oplus x_2 \oplus x_3 \oplus x_5 = 0$ and $x_3 \oplus x_4 \oplus x_6 = 1$. Let  $\sigma_a(1) = 4$, $\sigma_a(2) = 5$, $\sigma_a(4) = 1$, $\sigma_a(5) = 6$, $\sigma_a(6) = 2$ and $\sigma_b(1) = 1$, $\sigma_b(2) = 5$, $\sigma_b(3) = 2$, $\sigma_b(5) = 3$. Then the graph $H$ is:
$$
H = \raisebox{-0.5\height}
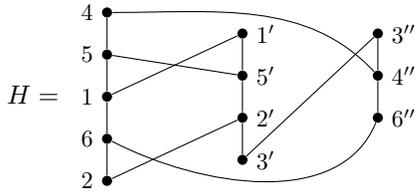
{
	\begin{tikzpicture}[xscale=0.9, yscale=0.8, every node/.style={scale=0.9}]
		\def\x{2};
		\def\y{0.7};
		\def\s{1.5};
				
		\node[circle,fill=black, inner sep=\s, label=left:{$4$}] (5) at (0,0) { };
		\node[circle,fill=black, inner sep=\s, label=left:{$5$}] (8) at (0,-1*\y) { };
		\node[circle,fill=black, inner sep=\s, label=left:{$1$}] (1) at (0,-2*\y) { };
		\node[circle,fill=black, inner sep=\s, label=left:{$6$}] (7) at (0,-3*\y) { };
		\node[circle,fill=black, inner sep=\s, label=left:{$2$}] (2) at (0,-4*\y) { };

		\node[circle,fill=black, inner sep=\s,label=right:{$3''$}] (4'') at (2*\x,-0.5*\y) { };
		\node[circle,fill=black, inner sep=\s, label=right:{$4''$}] (5'') at (2*\x,-1.5*\y) { };
		\node[circle,fill=black, inner sep=\s, label=right:{$6''$}] (7'') at (2*\x,-2.5*\y) { };
				
		\node[circle,fill=black, inner sep=\s, label=right:{$1'$}] (1') at (1*\x,-0.5*\y) { };
		\node[circle,fill=black, inner sep=\s, label=right:{$5'$}] (8') at (1*\x,-1.5*\y) { };
		\node[circle,fill=black, inner sep=\s, label=right:{$2'$}] (2') at (1*\x,-2.5*\y) { };
		\node[circle,fill=black, inner sep=\s, label=right:{$3'$}] (4') at (1*\x,-3.5*\y) { };
				
		\draw (1) -- (1');
		\draw (8) -- (8');
		\draw (2) -- (2');
		\draw (4') -- (4'');
		\draw (7) to [out=-30, in=-110] (7'');
		\draw (5) to [out=0, in=130] (5'');
				
		\draw (5) -- (8) -- (1) -- (7) -- (2);
		\draw (1') -- (8') -- (2') -- (4');
		\draw (4'') -- (5'') -- (7'');
	\end{tikzpicture}
}$$
The Tseitin graph of $\rAddPar(a,b)$ is the graph $H$ above after contraction of the edges $(4,5)$, $(6,2)$, $(1',5')$, $(2',3')$, $(3'',4'')$ and $(4'',6'')$.
\end{example}

\iftrue
Let $n' = |A \cup B|$. Since $A \cap B$,  $A\cap(A\triangle B)$ and $B\cap (A \triangle B)$ are pairwise disjoint and form a partition of $A \cup B$, we have that $\max(|A\cap B|,|A\cap(A\triangle B)|,|B\cap (A \triangle B)|) \geq \frac{n'}{3}$. We assume, without loss of generality, that $|A \cap B| \geq \frac{n'}{3}$. When both $A$ and $B$ are chosen uniformly at random from $X$ then, with high probability, we have that $n' \geq \frac{n}{2}$ and thus $|A \cap B| \geq \frac{n}{6}$. We now show that when $|A \cap B| > 0$ we can find a minor of $H$ that is isomorphic to a graph $G_\sigma$ with $\sigma \in \mathfrak{S}_{|A \cap B|}$. The minor, that we call $m(H)$, is obtained as follows:
\begin{itemize}
	\item[1.] remove all edges incident to a vertex of $V''$, then remove~$V''$ 
	\item[2.] contract every vertex $i \in V$ not connected to any $j' \in V'$ with one of its neighbors, repeat until $V$ contains only vertices connected to some vertex in $V'$
	\item[3.] contract every vertex $i' \in V'$ not connected to any $j \in V$ with one of its neighbors, repeat until $V'$ contains only vertices connected to some vertex in $V$
\end{itemize}
The construction is illustrated by Figure~\ref{fig:minor} for the graph $H$ of Example~\ref{example:3_layer_graph}. 
The construction gives the same minor regardless of the neighbors chosen for the contractions in steps 2 and~3. 

	\begin{figure}[h]
		\centering
		\begin{tikzpicture}[scale=0.8, xscale=0.6, every node/.style={scale=0.8}]
			\def\x{2};
			\def\y{0.7};
			\def\s{1.5};
			
			\node[circle,fill=black, inner sep=\s, label=left:{$4$}] (5) at (0,0) { };
			\node[circle,fill=black, inner sep=\s, label=left:{$5$}] (8) at (0,-1*\y) { };
			\node[circle,fill=black, inner sep=\s, label=left:{$1$}] (1) at (0,-2*\y) { };
			\node[circle,fill=black, inner sep=\s, label=left:{$6$}] (7) at (0,-3*\y) { };
			\node[circle,fill=black, inner sep=\s, label=left:{$2$}] (2) at (0,-4*\y) { };

			\node[circle,fill=black, inner sep=\s, label=right:{$1'$}] (1') at (1*\x,-0.5*\y) { };
			\node[circle,fill=black, inner sep=\s, label=right:{$5'$}] (8') at (1*\x,-1.5*\y) { };
			\node[circle,fill=black, inner sep=\s, label=right:{$2'$}] (2') at (1*\x,-2.5*\y) { };
			\node[circle,fill=black, inner sep=\s, label=right:{$3'$}] (4') at (1*\x,-3.5*\y) { };
			
			\draw (1) -- (1');
			\draw (8) -- (8');
			\draw (2) -- (2');
			
			\draw (5) -- (8) -- (1) -- (7) -- (2);
			\draw (1') -- (8') -- (2') -- (4');
			
			\node[align=center] (l) at (0.5*\x,-5*\y) {Step 1}; 
			
			node/.style={scale=0.8}]
			\def\x{2};
			\def\y{0.7};
			\def\s{1.5};
			
			\node[circle,fill=black, inner sep=\s, label=left:{$5$}] (8a) at (5,-1*\y) { };
			\node[circle,fill=black, inner sep=\s, label=left:{$1$}] (1a) at (5,-2*\y) { };
			\node[circle,fill=black, inner sep=\s, label=left:{$2$}] (2a) at (5,-3*\y) { };

			\node[circle,fill=black, inner sep=\s, label=right:{$1'$}] (1a') at (5+1*\x,-0.5*\y) { };
			\node[circle,fill=black, inner sep=\s, label=right:{$5'$}] (8a') at (5+1*\x,-1.5*\y) { };
			\node[circle,fill=black, inner sep=\s, label=right:{$2'$}] (2a') at (5+1*\x,-2.5*\y) { };
			\node[circle,fill=black, inner sep=\s, label=right:{$3'$}] (4a') at (5+1*\x,-3.5*\y) { };
			
			\draw (1a) -- (1a');
			\draw (8a) -- (8a');
			\draw (2a) -- (2a');
			
			\draw (8a) -- (1a) -- (2a);
			\draw (1a') -- (8a') -- (2a') -- (4a');
			
			\node[align=center] (la) at (5+0.5*\x,-5*\y) {Step 2};  
			
			node/.style={scale=0.8}]
			\def\x{2};
			\def\y{0.7};
			\def\s{1.5};
			
			\node[circle,fill=black, inner sep=\s, label=left:{$5$}] (8b) at (10,-1*\y) { };
			\node[circle,fill=black, inner sep=\s, label=left:{$1$}] (1b) at (10,-2*\y) { };
			\node[circle,fill=black, inner sep=\s, label=left:{$2$}] (2b) at (10,-3*\y) { };

			\node[circle,fill=black, inner sep=\s, label=right:{$1'$}] (1b') at (10+1*\x,-1*\y) { };
			\node[circle,fill=black, inner sep=\s, label=right:{$5'$}] (8b') at (10+1*\x,-2*\y) { };
			\node[circle,fill=black, inner sep=\s, label=right:{$2'$}] (2b') at (10+1*\x,-3*\y) { };
			
			\draw (1b) -- (1b');
			\draw (8b) -- (8b');
			\draw (2b) -- (2b');
			\draw (8b) -- (1b) -- (2b);
			\draw (1b') -- (8b') -- (2b') ;

			\node[align=center] (lb) at (10+0.5*\x,-5*\y) {Step 3};  
		\end{tikzpicture}
		\caption{Construction of $m(H)$}\label{fig:minor}
	\end{figure}
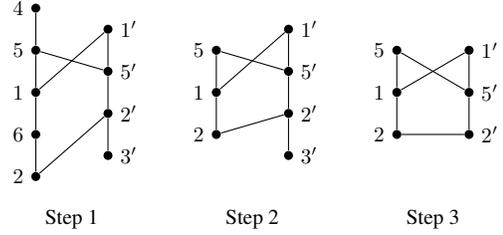

\begin{lemma}\label{lemma:obvious_lemma}
	When $A$, $B$, $\sigma_a$ and $\sigma_b$ are chosen independently and uniformly at random then, conditioned on $|A \cap B| = k$, the minor $m(H)$ follows the same distribution as $G_\sigma$ when $\sigma$ is sampled uniformly at random in $\mathfrak{S}_k$.
\end{lemma}
\begin{proof}
	Consider the event ``$m(H)$ is isomorphic to $G_\sigma$'', which we just denote by $m(H) = G_\sigma$. Fix $A$ and $B$. For every $\sigma_a$ there is a unique  bijection $\rho_a : A \cap B \rightarrow [k]$ such that $\rho_a^{-1} \circ \sigma_a \circ \rho_a = \id \in \mathfrak{S}_k$ ($\rho_a$ is just a renaming of the vertices). Let $\sigma \in \mathfrak{S}_k$ be the unique permutation such that $m(H)$ is isomorphic to $G_\sigma$, then $\sigma = \rho_a \circ \sigma_b \circ \sigma_a^{-1} \circ \rho_a^{-1}$. For every $\sigma \in \mathfrak{S}_k$, there exists the same number of pairs $(\sigma_a,\sigma_b)$ such that $\sigma = \rho_a \circ \sigma_b \circ \sigma_a^{-1} \circ \rho_a^{-1}$, so when $A$ and $B$ are fixed, the probability that $m(H) = G_\sigma$ is the same for all $\sigma$. This probability is independent of $A$ and $B$ as long as $|A \cap B| = k$, so the result follows.
\end{proof}

\begin{lemma}\label{lemma:3layers_high_tw}
	There is a constant $\alpha > 0$ such that, when $A$, $B$, $\sigma_a$ and $\sigma_b$ are chosen independently and uniformly at random, the probability that $tw(H) < \alpha n$ vanishes to $0$ as $n$ goes to infinity.
\end{lemma}
\begin{proof}
	Suppose, without loss of generality, that we always have $\max(|A\cap B|,|A\cap(A\triangle B)|,|B\cap (A \triangle B)|) = |A \cap B|$. Fix $k \in \mathbb{N}$. By Lemmas~\ref{lemma:expander_whp} and~\ref{lemma:obvious_lemma}, there is a constant $\beta$ such that $\textup{Pr}(\tw(m(H)) < \beta k \mid |A \cap B| = k)$ tends to $0$ as $k$ increases. Since $m(H)$ is a minor of $H$, we have that $\tw(m(H)) \leq \tw(H)$, thus $\textup{Pr}(\tw(H) < \beta k \mid |A \cap B| = k)$ also tends to $0$ as $k$ increases. It follows that $\textup{Pr}(\tw(H) < \frac{\beta n}{6} \mid |A \cap B| \geq \frac{n}{6})$ vanishes to $0$ as $n$ increases, and since $\textup{Pr}(|A \cap B| \geq \frac{n}{6})$ tends to $1$ as $n$ increases, we finally obtain that $\textup{Pr}(\tw(H) < \frac{\beta n}{6})$ goes to $0$ as $n$ increases.
\end{proof}

The graphs $H$ have degree at most $3$ and by Lemma~\ref{lemma:3layers_high_tw} they have linear treewidth with high probability, so using Theorem~\ref{theorem:linear_tw_res} we immediately have that, with high probability, Tseitin formulas over $H$ are hard for resolution refutation. Given the relation between $H$ and $\rAddPar(a,b)$'s Tseitin graphs described in Lemma~\ref{lemma:tw_after_contraction_2}, we conclude that almost all formulas $\rAddPar(a,b)$ are hard for resolution.
\fi

\iffalse
When the constraints $a$ and $b$ are constructed in a random fashion, the hardness of $\rAddPar(a,b)$ for resolution stems from the hardness of $\rPar$. One proves this by showing that, when the parameters of $a$ and $b$ are sampled uniformly at random, $H$ is a random graph that, w.h.p., admits a graph $G_\sigma$ as a minor, for $\sigma$ a permutation over $\Omega(n)$ elements of $X$. When this minor exists (which is almost always the case), it also follows a uniform distribution, and thus its treewidth is in $\Omega(n)$ by Lemma~\ref{lemma:expander_whp}. This then shows that $\tw(H) = \Omega(n)$ holds w.h.p., and the lower bounds for resolution follows.
\fi

\begin{theorem}\label{theorem:hard_formulas_3layers}
	There is a constant $\alpha > 0$ such that, with probability tending to $1$ as $n$ increases, when $A$, $B$, $\sigma_a$ and $\sigma_b$ are chosen independently and uniformly at random, the length of a shortest resolution refutation of $\rAddPar(a,b)$ is least~$2^{\alpha n}$.  
\end{theorem} 
\iftrue
\begin{proof}
$\rAddPar(a,b)$ is Tseitin graph whose graph $G$ has maximum degree at most $4$. By Lemma~\ref{lemma:rAddPar_and_H} we have $\tw(G) \geq \tw(H) - 6$ so, by Lemma~\ref{lemma:3layers_high_tw}, $\tw(G) = \Omega(n)$ with high probability when $A$, $B$, $\sigma_a$ and $\sigma_b$ are chosen independently and uniformly. The statement of the theorem then follows from Theorem~\ref{theorem:linear_tw_res}.
\end{proof}
\fi

\label{sec:fav}
Notice here that $\sigma_a$, $\sigma_b$ and $\id$ are relative to each other shuffled randomly. This is the most chaotic scenario which is likely to contribute to its difficulty. Let us instead briefly discuss an example that favors shorter proofs. We are given $a$ and $b$ in their random orders but, in an addition step, we may be the ones creating the encoding for the sum constraint and so we can choose the permutation $\sigma_c$ favorably by ensuring that $\sigma_c(i)< \sigma_c(j)$ if and only if either:
\begin{itemize}
	\item $x_i\in A \setminus B$ and $x_j \in B \setminus A $
	\item $x_i, x_j \in A \setminus B$ and  $\sigma_a(i)< \sigma_a(j)$
	\item $x_i, x_j \in B \setminus A$ and  $\sigma_b(i)< \sigma_b(j)$
\end{itemize}
\noindent 
Even in this case, the encoding $\rAddPar(a,b,\sigma_c) = a \land b \land \parity(C,U,\sigma_c)$ will be hard w.h.p. when $|A \cap B| = \Omega(n)$ since then we can find the minor $G_\sigma$ evoked above by only looking in $H[V \cup V']$. The condition $|A \cap B| = \Omega(n)$ is fulfilled almost surely when $A$ and $B$ are chosen uniformly. 

\section{Sorting and Upper Bounds}

\citet{PR16} showed that XOR-reasoning such as in Gaussian elimination
can have short proofs; a BDD approach can find polynomial-size
extended resolution proofs \cite{BDDER}. For these particular formulas
we can do even better, reducing the complexity and the number of extra
variables needed.

\begin{lemma}[\citealt{ChewH20}]\label{switching}
	Suppose we have a CNF $F$ and two sets of XOR clauses $\xor(x,y,p)$
	%\todo{some renaming is needed here ($a$ and $b$)}
	and $\xor(p,z,q)$, where variable $p$ appears nowhere in $F$. We can infer:
	\begin{prooftree}
		\AxiomC{$F \wedge\xor(x,y,p)\wedge\xor(p,z,q)$}
		\UnaryInfC{$F \wedge\xor(y,z,p)\wedge\xor(p,x,q)$}
	\end{prooftree}
	in 32 of DRAT steps without adding new variables.
	
\end{lemma}
\iftrue
\begin{proof}
	We can do this in 32  steps as in Figure~\ref{fig:switch}.
	\begin{figure}[h!]
		\centering
		\text{ }\hfill 
		\begin{minipage}{0.10\textwidth}
			\centering
			ATA
			
			\smallskip
			
			\framebox{
				\parbox{0.95\textwidth}{
					$$
						\bar q\vee x \vee y \vee z $$ $$\bar q\vee \bar x \vee \bar y \vee z$$ 
						$$\bar q\vee x \vee \bar y \vee \bar z $$ $$ \bar q\vee \bar x \vee y \vee \bar z$$ $$
						q\vee \bar x \vee y \vee z $$ $$  q\vee x \vee  \bar y \vee  z$$ $$
						q\vee x \vee  y \vee \bar z $$ $$  q\vee \bar x \vee \bar y \vee \bar z$$
					%\vspace{-10pt}
					
			}}
		\end{minipage}\hfill
		\begin{minipage}{0.12\textwidth}
			\centering
			RATE
			
			\smallskip
			
			\framebox{
				\parbox{0.95\textwidth}{
					$$
						\mathtt{d} \quad \bar p\vee x \vee y  $$ $$
						\mathtt{d} \quad \bar p\vee \bar x \vee \bar y  $$ $$
						\mathtt{d} \quad p\vee x \vee \bar y  $$ $$
						\mathtt{d} \quad p\vee \bar x \vee y  $$ $$
						\mathtt{d} \quad  \bar p\vee q \vee z  $$ $$
						\mathtt{d} \quad  \bar p\vee \bar q \vee \bar z $$ $$
						\mathtt{d} \quad    p\vee  q \vee \bar z$$ $$
						\mathtt{d} \quad   p\vee  \bar q \vee  z $$
			}}
		\end{minipage}\hfill
		\begin{minipage}{0.10\textwidth}
			\centering
			RATA
			
			\smallskip
			
			\framebox{\parbox{0.95\textwidth}{
				
						$$\bar p\vee z \vee y  $$
						$$\bar p\vee \bar z \vee \bar y  $$
						$$p\vee z \vee \bar y  $$
						$$p\vee \bar z \vee y  $$ 
						$$\bar p\vee q \vee x  $$
						$$\bar p\vee \bar q \vee \bar x $$
						$$p\vee  q \vee \bar x$$
						$$p\vee  \bar q \vee  x $$
			}}
		\end{minipage}\hfill \text{ }
	
		\begin{minipage}{0.25\textwidth}
			\centering
			\smallskip 
			
			ATE		
			\smallskip 	
			
			\framebox{\parbox{1\textwidth}{
			$$
						\mathtt{d} \quad \bar q\vee x \vee y \vee z \\ \quad\mathtt{d} \quad  \bar q\vee \bar x \vee \bar y \vee z $$
						$$\mathtt{d} \quad \bar q\vee x \vee \bar y \vee \bar z \\ \quad\mathtt{d} \quad  \bar q\vee \bar x \vee b \vee \bar z $$
						$$\mathtt{d} \quad q\vee \bar x \vee y \vee z \\ \quad\mathtt{d} \quad   q\vee x \vee  \bar y \vee  z$$
						$$\mathtt{d} \quad q\vee x \vee  y \vee \bar z \\ \quad\mathtt{d} \quad   q\vee \bar x \vee \bar y \vee \bar z $$
			}}
		\end{minipage}
		\caption{DRAT steps required for Lemma~\ref{switching}, \texttt{d} denotes a deletion step. \label{fig:switch}}
	\end{figure}
\end{proof}

\fi
\iftrue
\begin{lemma}[\citealt{ChewH20}]\label{switchingER}
	Suppose we have a CNF $F$ and two sets of XOR clauses $\xor(x,y,p)$ and $\xor(p,z,q)$, where variable $p$ appears nowhere in $F$. We can start with $F \wedge\xor(x,y,p)\wedge\xor(p,z,q)$ and by introducing a new variable $p'$ infer all clauses from $\xor(y,z,p')$ and $\xor(p',x,q)$ in a constant number of Extended Resolution steps.
\end{lemma}

\begin{proof}
	We emulate the DRAT proof except we remove all deletion steps, derive the ATA clauses by resolution, and the $\xor(y,z,p')$ clauses can be introduced as extension clauses on $p'$. Then the $\xor(p',x,q)$ clauses can be derived by resolution with the ATA clauses. 
\end{proof}
\fi

\begin{lemma}[\citealt{ChewH20}] \label{reordering}
	Given two permutations $\sigma_1$ and  $\sigma_2$, $X$, $S$ and $T$ are disjoint sets of variables where both $S$ variables and $T$ variables do not appear in CNF $F$, $F\wedge\parity(X,S,\sigma_1)$ can be transformed into  $F\wedge\parity(X,T,\sigma_2)$ in $O(n\log n)$ many DRAT steps where $|X|=n$.
\end{lemma}
\begin{proof}[Sketch Proof]
	\begin{enumerate}
		\item In $O(n \log n)$ applications of Lemma~\ref{switching} we can take the linear structure of the Tseitin variables and reorganize it into a balanced binary tree, using a divide-and-conquer approach. 
		\item In $O(n \log n)$ applications of Lemma~\ref{switching} we can make any permutation of the leaf edges. 
		By first swapping any two leaves takes $O(\log n)$ applications of Lemma~\ref{switching}
		and then $n-1$ swaps are required (in the worst case) and sufficient to place every variable in place. 
		\item In $O(n \log n)$ applications of
                  Lemma~\ref{switching} we can take our balanced
                  binary tree and return it to a linear structure. \qedhere
	\end{enumerate}
\end{proof}
\iftrue
\begin{theorem}[\citealt{ChewH20}]\label{thm:DRATrpar}
	For any permutation $\sigma \in \mathfrak{S}_n$, there are DRAT$^-$ refutations of $\rPar(n, \sigma)$ that have $O(n\log n)$ many lines with no new variables.
\end{theorem}
\fi
\iftrue
\begin{proof}[Sketch Proof]
	\begin{enumerate}
		\item Using Lemma~\ref{reordering}, we can rearrange the two parity constraints to be in the same ordering. 
		\item With this ordering, the treewidth of the Tseitin graph shrinks, and a linear induction proof showing the equivalence of the Tseitin variables in resolution can be performed. 
	\end{enumerate}
\end{proof}
\fi

\iftrue
\begin{theorem}[\cite{ChewH20}]
	For any permutation $\sigma \in \mathfrak{S}_n$, there are Extended Resolution refutations of $\rPar(n, \sigma)$ that have $O(n\log n)$ many lines with no new variables.
\end{theorem}

\begin{proof}
	This is a corollary of Theorem~\ref{thm:DRATrpar}. Lemma~\ref{reordering} can be done entirely in Extended Resolution using Lemma~\ref{switchingER}. The remainder of the proof is a resolution proof so is unaffected by that change from DRAT to Extended Resolution. 
\end{proof}
\fi

\begin{theorem}\label{thm:threedrat}
	For any parity constraints $a,b $ over $n$ input variables $X$, there are DRAT$^-$ refutations of $\rAddPar(a,b)$ that have $O(n\log n)$ many lines.
\end{theorem}
\iffalse
\begin{proof}[Sketch Proof]
	Using Lemma~\ref{reordering} we can rearrange the the random orderings to an easy case where all variables appear in order in $O(n\log n)$-step. The remaining refutation is a linear-sized resolution refutation. 
\end{proof}
\fi

\iftrue
\begin{proof}
	If both $a$ and $b$ have the same variable set either this reduces to a case in Theorem~\ref{thm:DRATrpar}, or it trivially includes the empty clause as the final constraint.
	
	Otherwise, without loss of generality, assume both $a$ and $b$ express even parities, we can flip a literal to achieve this.
	$c$ is the third constraint in $\rAddPar(a,b)$.
	Let $X$ be the set $\SB x_i \SM 0< i \leq n\SE$, without loss of generality assume each $x_i$ appears in at least one of $a$ or $b$. 
	$$a =  \xor(x_{\sigma_a(1)}, x_{\sigma_a(2)}, t^a_1)\wedge {\bigwedge_{j=1}^{n_a-4}  \xor(t^a_{j}, {x_{\sigma_a(j+2)}},t^a_{j+1} )}$$ $$\wedge \xor({t}^a_{n_a-3},x_{\sigma_a(n_a-1)}, x_{\sigma_a(n_a)}) $$
	$$b =  \xor(x_{\sigma_b(1)}, x_{\sigma_b(2)}, t^b_1)\wedge {\bigwedge_{j=1}^{n_b-4}  \xor(t^b_{j}, {x_{\sigma_b(j+2)}},t^b_{j+1} )}$$ $$\wedge \xor({t}^b_{n_b-3},x_{\sigma_b(n_b-1)}, x_{\sigma_b(n_b)})$$
	$$ c =  \xor(x_{\sigma_c(1)}, x_{\sigma_c(2)}, t^c_1)\wedge {\bigwedge_{j=1}^{n_c-4}  \xor(t^c_{j}, {x_{\sigma_c(j+2)}},t^c_{j+1} )}$$ $$\wedge \xor(\neg {t}^c_{n_c-3},x_{\sigma_c(n_c-1)}, x_{\sigma_c(n_c)}) $$
	
	We first consider the ordered special case when $\sigma_a, \sigma_b, \sigma_c $ all preserve order. E.g. $\sigma_l(i)< \sigma_l(j)$ for $i< j$ and $l\in \{a,b,c\}$.
	For $l\in \{a,b,c\}$ and $1 <j\leq n-4$, we find the highest $i\leq j+1$ such that there is some $k$ so that $i=\sigma_l(k)$, since $\sigma_l$ is injective, $\tau_l(j)$ is defined as $k-1$. 
	The purpose of $\tau_l$, is that $t^l_{\tau_l(j)}$ is defined as parity of all the variables in $l$ up to $x_{j+1}$ by index, this way we can talk about comparable $t$-variables from $a,b$ and $c$ together.\medskip

	\textbf{Induction Hypothesis:}
	Suppose 
	\begin{multline*}
	\max (\sigma_a(2),\sigma_b(2),\sigma_c(2 ))-1 < j \\ < \min (\sigma_a(n_a-1),\sigma_b(n_b-1),\sigma_c(n_c-1) )-1
	\end{multline*} then all clauses of $\xor(t^a_{\tau_a(j)},t^b_{\tau_b(j)}, t^c_{\tau_c(j)}) $ can be proven in an $O(j)$ size resolution proof.\medskip
	
	\textbf{Inductive Step:}
	Starting with the four clauses from $\xor(t^a_{\tau_a(j)},t^b_{\tau_b(j)}, t^c_{\tau_c(j)}) $, we introduce $x_{j+2}$ and eliminate it. Because each $x$ appears in two constraints. There are $u, v\in \{a,b,c\}$ such that $u\neq v$ and $\xor(t^{u}_{\tau_u(j)},x_{j+2}, t^{u}_{\tau_u(j)+1}) $ and $\xor(t^{v}_{\tau_v(j)},x_{j+2}, t^{v}_{\tau_v(j)+1}) $ are part of $\rAddPar(a,b)$. Assume without loss of generality $u=a$, $v=b$
	
	First we can eliminate $t^a_{\tau_a(j)}$, we get 8 clauses, 
	then for each of the eight clauses there is a unique clause in $\xor(t^{b}_{\tau_b(j)},x_{j+2}, t^{b}_{\tau_b(j)+1}) $ for which it can resolve without a tautological resolvent.

	\begin{figure}[h]\centering	
		\begin{minipage}{0.35\textwidth}
			\centering

			\smallskip
			
			\framebox{
				\parbox{0.85\textwidth}{

					$$
						\bar t^b_{\tau_b(j)} \vee t^c_{\tau_c(j)} \vee x_{j+2} \vee t^{a}_{\tau_a(j)+1} $$ $$
						\bar t^b_{\tau_b(j)}\vee \bar t^c_{\tau_c(j)} \vee \bar x_{j+2} \vee t^{a}_{\tau_a(j)+1}$$ $$
						\bar t^b_{\tau_b(j)}\vee t^c_{\tau_c(j)} \vee \bar x_{j+2} \vee \bar t^{a}_{\tau_a(j)+1} $$ $$ \bar t^b_{\tau_b(j)}\vee \bar t^c_{\tau_c(j)} \vee x_{j+2} \vee \bar t^{a}_{\tau_a(j)+1}$$ $$
						t^b_{\tau_b(j)}\vee \bar t^c_{\tau_c(j)} \vee x_{j+2} \vee t^{a}_{\tau_a(j)+1} $$ $$ t^b_{\tau_b(j)}\vee t^c_{\tau_c(j)} \vee  \bar x_{j+2} \vee  t^{a}_{\tau_a(j)+1}$$ $$
						t^b_{\tau_b(j)}\vee t^c_{\tau_c(j)} \vee  x_{j+2} \vee \bar t^{a}_{\tau_a(j)+1} $$ $$  t^b_{\tau_b(j)}\vee \bar t^c_{\tau_c(j)} \vee \bar x_{j+2} \vee \bar t^{a}_{\tau_a(j)+1}$$
			}}
		\end{minipage}	
		\begin{minipage}{0.35\textwidth}
			\centering

			\smallskip
			
			\framebox{
				\parbox{0.85\textwidth}{

					$$
						\bar t^b_{\tau_b(j)+1} \vee t^c_{\tau_c(j)} \vee x_{j+2} \vee t^{a}_{\tau_a(j)+1} $$ $$ t^b_{\tau_b(j)+1}\vee \bar t^c_{\tau_c(j)} \vee \bar x_{j+2} \vee t^{a}_{\tau_a(j)+1}$$ $$
						t^b_{\tau_b(j)+1}\vee t^c_{\tau_c(j)} \vee \bar x_{j+2} \vee \bar t^{a}_{\tau_a(j)+1} $$ $$ \bar t^b_{\tau_b(j)+1}\vee \bar t^c_{\tau_c(j)} \vee x_{j+2} \vee \bar t^{a}_{\tau_a(j)+1}$$ $$
						t^b_{\tau_b(j)+1}\vee \bar t^c_{\tau_c(j)} \vee x_{j+2} \vee t^{a}_{\tau_a(j)+1} $$ $$  \bar t^b_{\tau_b(j)+1}\vee t^c_{\tau_c(j)} \vee  \bar x_{j+2} \vee  t^{a}_{\tau_a(j)+1}$$ $$
						t^b_{\tau_b(j)+1}\vee t^c_{\tau_c(j)} \vee  x_{j+2} \vee \bar t^{a}_{\tau_a(j)+1} $$ $$ \bar t^b_{\tau_b(j)+1}\vee \bar t^c_{\tau_c(j)} \vee \bar x_{j+2} \vee \bar t^{a}_{\tau_a(j)+1}$$

			}}
		\end{minipage}
	\end{figure}
	
	The next 8 clauses can resolve in pairs over $x_{j+2}$ to produce the four clauses of $\xor(t^a_{\tau_a(j)+1},t^b_{\tau_b(j)+1}, t^c_{\tau_c(j)}) $:
	$$
		\bar t^b_{\tau_b(j)+1} \vee t^c_{\tau_c(j)} \vee t^{a}_{\tau_a(j)+1} \qquad   t^b_{\tau_b(j)+1}\vee \bar t^c_{\tau_c(j)}  \vee t^{a}_{\tau_a(j)+1}$$ $$
		t^b_{\tau_b(j)+1}\vee t^c_{\tau_c(j)}  \vee \bar t^{a}_{\tau_a(j)+1} \qquad  \bar t^b_{\tau_b(j)+1}\vee \bar t^c_{\tau_c(j)} \vee \bar t^{a}_{\tau_a(j)+1}
	$$
	
	${\tau_a(j)+1}=\tau_a(j+1)$, $\tau_b(j)+1=\tau_b(j+1)$, and $\tau_c(j)= \tau_c(j+1)$. So we are done in $+20$ steps ($+48$ if we want to include deletions) and incremented $j$ by 1.\medskip

\begin{figure*}[ht]
	\centering
	\input{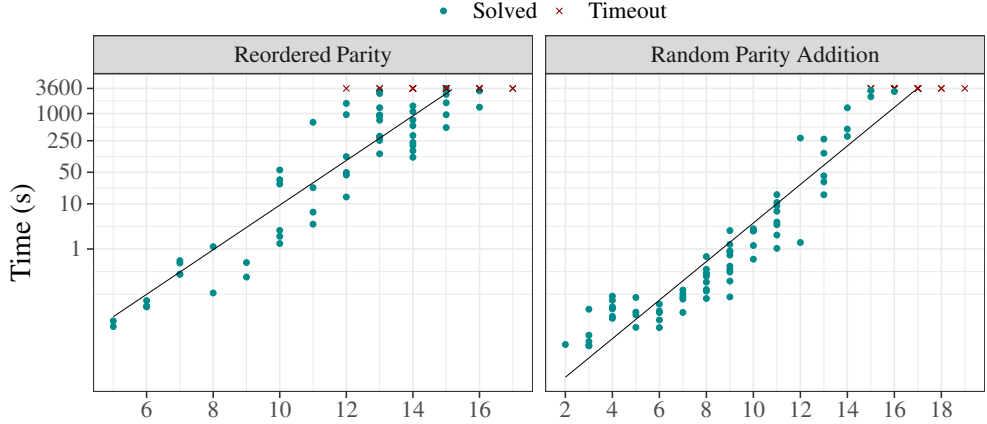}
	\caption{Treewidth (x-axis) vs. solving time (y-axis) for reordered parity (left) and random parity addition (right).}\label{fig:treewidth-sat}
\end{figure*}

	\textbf{Base Case:}
	In the inductive step we proceeded by eliminating two $t$ variables then resolving an $x$ variable. 
	We can do the same with the base case, however we do not start with a clause with three $t$ variables, but the 12 clauses of $\xor(x_{\sigma_a(1)}, x_{\sigma_a(2)}, t^a_1)$, $\xor(x_{\sigma_b(1)}, x_{\sigma_b(2)}, t^b_1)$, $\xor(x_{\sigma_c(1)}, x_{\sigma_c(2)}, t^c_1)$. 
	We would start by resolving any pair of $x$ that appear across these constraints, after that we introduce new $x$ variables by eliminating a pair of $t$ variables and then resolving away the $x$ variables.
	We want to eliminate all $x_i$ for $i \leq j+1$ in increasing order.

	Immediately before eliminating $x_i$ we eliminate all of $t^a_{\tau_a(i-2)},t^b_{\tau_b(i-2)}, t^c_{\tau_c(i-2)}$ if any of them exist have not been eliminated prior.
	In other words we eliminate the smallest index $t$ variables for clauses that include $x_i$. There will be exactly two %prove
	of these that have not been eliminated prior. After that we can resolve on $x_i$.
	
	Effectively, this sums three parity constraints together, but without extension variables the number of clauses is $2^k$ where $k$ is the length of the parity constraint. Our order matters here because we want to keep $k$ bounded. 
	Note that we have a maximum of $3$ $t$ variables in each working clause, as we only introduce a new one when eliminating an old one. We also introduce a new $x_i$ into the working clauses, when resolving on a $t$ variable. 
	However directly before $t$ is resolved upon, $x_{i-1}$ is eliminated in our procedure. Therefore the number of $x$ variables is also bounded by a constant
	
	We can see this bound as being $6$ as we start with the initial 
	$x$ variables: $x_{\sigma_a(1)}$, $x_{\sigma_a(2)}$, $x_{\sigma_b(1)}$, $x_{\sigma_b(2)}$, $x_{\sigma_c(1)}$, $x_{\sigma_c(2)}$, and we only introduce a new variable after eliminating an older one.
	In fact the bound is at most $4$ because there are at least two pairs (one pair for each of the two lowest values) in this set of identical values. So the width is $7$ or lower.

	We reach the base case once we have eliminated all of $x_{\sigma_a(1)}, x_{\sigma_a(2)}, x_{\sigma_b(1)}, x_{\sigma_b(2)}, x_{\sigma_c(1)}, x_{\sigma_c(2)}$.\medskip

	\textbf{Reaching a contradiction:}
	Reaching the contradiction is similar to the base case of the induction, we simply work in reverse from the other sides of the parity constraints. 
	By the induction proof we reach $\xor(t^a_{\tau_a(j)},t^b_{\tau_b(j)}, t^c_{\tau_c(j)}) $, but end case will generate $\neg \xor(t^a_{\tau_a(j)},t^b_{\tau_b(j)}, t^c_{\tau_c(j)}) $ because of the flipped literal in the $c$ constraint. We can end with a few resolution steps to get the empty clause.\medskip

	\textbf{Reducing to this special case:}
	Using Lemma~\ref{reordering} we can turn any other case to this special case in $O(n\log n)$ many lines. 
\end{proof}

\fi

\iftrue
Once again, these easily convert into Extended Resolution proofs
\begin{theorem}\label{thm:threeER}
	For any parity constraints $a,b $ over $n$ input variables $X$, there are Extended Resolution refutations of $\rAddPar(a,b)$ that have $O(n\log n)$ many lines.
\end{theorem}

\begin{proof}
	We take the proof of Theorem~\ref{thm:threedrat}, every use of RATA comes about from a use of Lemma~\ref{switching} in the reordering process, which we can replace with steps from Lemma~\ref{switchingER}.
\end{proof}
\fi

\section{Experiments}\label{sec:experiments}

We ran experiments to confirm that the reordered parity and random parity addition formulas are hard to refute for CDCL solvers, and that their hardness is largely explained by the treewidth of their Tseitin graphs.
This is expected given the lower bounds of Theorems~\ref{theorem:hard_formulas} and~\ref{theorem:hard_formulas_3layers}, but these results are asymptotic and probabilistic, and it is not certain that they apply to relatively small formulas encountered in practice.

The experiments described here were performed on a cluster with Intel Xeon E5649 processors at 2.53~GHz running 64-bit Linux. An 8~GB memory limit and varying time limits were enforced with {\sc RunSolver}~\cite{Roussel11}.

\subsection{Problem 1: Reordered Parity}
We generated a benchmark set of  $\rPar(n, \sigma)$ for $n = 50$. Experiments for increasing values of $n$ were done by ~\citet{ChewH20}. To get formulas of varying treewidth, the permutations $\sigma$ were constructed in several ways:
\begin{enumerate}[left=0pt,label=(\alph*)]
	\item $5$ permutations were drawn from a uniform distribution.
	\item $30$ permutations were obtained from a stochastic process following a Mallows distribution 
 	whose parameter controls the likelihood of inversions~\cite{mallows1957}. 
	\iftrue Prior work on the distribution tends to show that the treewidth of ``tangled path'' graphs similar to our contracted Tseitin graphs increases with this parameter~\cite{EnrightMPS21}. The stochastic process is known as $q$-Mallows process and is described in~\cite[Section 2]{BhatnagarP15}.\fi
	\item $30$ permutations come from a sequence of random adjacent swaps, with a varying number of swaps.
	\item $15$ permutations were constructed by a sequence of random adjacent swaps until an element is a set distance away from its original position in the original order.
\end{enumerate}
An ideal construction would have allowed us to uniformly sample graphs $G_\sigma$ for a fixed $n$, with treewidth lying in a fixed range. But we know of no such construction that is also efficient. Hence the three last constructions listed above, that have parameters that intuitively give us some control on the treewidth. The drawback is that the graphs are not sampled uniformly at random, as in the theoretical results.

The resulting Tseitin graphs have $100$ vertices each, so determining their treewidth is challenging.
We obtained upper and lower bounds using tools by \citet{Tamaki22}.\footnote{\url{https://github.com/twalgor/tw}}
Within a time limit of $3600$ seconds, the treewidth of only~$30$ graphs could be computed exactly.
For another~$30$ graphs, no non-trivial lower bound was returned.
To get lower bounds for such graphs~$G_\sigma$, we determined the treewidth of the graph~$G^*_\sigma$ obtained by contracting the edges $(i,\sigma(i))$ for \hbox{$i \in n$}.
Lemma~\ref{lemma:tw_after_contraction} shows that $\tw(G^*_\sigma)
\leq \tw(G_\sigma) \leq 2 \cdot \tw(G^*_\sigma)$, and since the graphs $G^*_\sigma$ only contain half as many vertices, matching upper and lower bounds could be computed for all except two instances (curiously, these were graphs for which the treewidths of the original graphs $G$ could be determined quickly). 
\iftrue
Figure~\ref{fig:treewidthUL} shows the best lower and upper bounds on the treewidth obtained for each graph.
\fi

We ran the CDCL solver CaDiCaL~\cite{BiereFazekasFleuryHeisinger2020} on each reordered parity formula with default settings and a timeout of $3600$ seconds. CaDiCaL generates Reverse Unit Propagation (RUP) proofs~\cite{Gelder08,HeuleHW13}, which can be converted to resolution proofs with a quadratic overhead~\cite{GoldbergN03}.
Figure~\ref{fig:treewidth-sat} (left) plots the solver's running time against the best upper bounds on the treewidth for each instance.
The regression line clearly shows that the running time grows exponentially with the treewidth.
By contrast, CryptoMiniSat~\cite{SNC09}, a solver that is capable of reasoning with XORs using Gaussian elimination, was able to solve all instances within a few seconds.

\subsection{Problem 2: Random Parity Additions}
We generated a benchmark set consisting of $95$ reordered parity formulas $\rAddPar(a, b, \sigma_c)$. To get formulas of varying treewidth, we altered the value $p$, which dictates the probability of each input variable being selected to be including in $a$, and the same probability also for $b$. 
We chose~$p$ in increments of $0.05$ from $0.05$ up to $0.95$. We scaled the number of variables we drew from based on $p$ to keep the expectation of the number of clauses the same, in our sample the number of clauses ended up being between 360 and~536. $\sigma_c$ was chosen to be favorable to solving. 
Figure~\ref{fig:treewidth-sat} plots the solver's running time against the best upper bounds on the treewidth for each instance. Once again we see running time grows exponentially with the treewidth.
\iftrue
\begin{figure}
	\centering
	\input{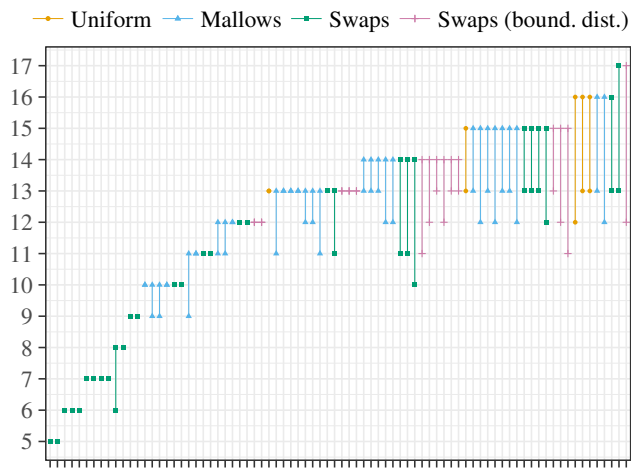}\caption{Lower and upper bounds for the treewidth (y-axis) of reordered parity formulas (x-axis).}\label{fig:treewidthUL}
\end{figure}
\fi
\section{Conclusion}

We present both theoretical and experimental evidence that treewidth explains the hardness of reordered parity, and random parity additions for CDCL/resolution. 

\citet{ChewH20} have left the DRAT$^-$ upper bound for $\rPar$ without
a proof lower bound for resolution.  We have now provided that. In
particular, noticing that the instances were Tseitin formulas, we were
driven to study the treewidth of their underlying graphs and we have
shown that it is, with high probability, linear in the number of
variables. And although the relationship between resolution
refutations of Tseitin formulas and the graph's treewidth is not fully
understood yet, results do exist for linear treewidth that are enough
for us to prove exponential lower bounds on the length of the
resolution proofs. Previous experiments showed the exponential
increase in CaDiCaL proof size as the number of variables increases~\cite{ChewH20},
in this paper we show that same exponential increase (but in solving
time), but with the number of variables and clauses controlled and now
the treewidth being varied.

We generalize this further to $\rAddPar$ which draws its motivation from Gaussian elimination. $\rAddPar$ provides yet another example of a hard Tseitin formula, and its hardness is confirmed both theoretically and experimentally. Again,  treewidth is the important factor in determining its hardness. In both the $\rPar$ and $\rAddPar$ case we can draw the conclusion that the variable order matters.

Just as in $\rPar$, we can show that we have short DRAT$^-$ proofs for $\rAddPar$. This will be useful for verification, with a hope that it may be generalized to Gaussian elimination. 

In future work, it would be interesting to explore the BDD-techniques for dealing with XOR-constraints, in a similar manner to our exploration on CDCL here.
On reordered parity EBDDRES \cite{BDDER} can perform even more poorly than CaDiCaL, but recent work \cite{BBH22} show BDD-solvers can perform even better than Chew and Heule's sorting tool, so the overall picture may be more complicated.

\section{Acknowledgments}

The authors acknowledge the support from the FWF (P36420, ESP 197, ESP 235) and the WWTF (ICT19-060, ICT19-065).

\bibliography{main}

\begin{thebibliography}{38}
\providecommand{\natexlab}[1]{#1}

\bibitem[{Alekhnovich and Razborov(2011)}]{AlekhnovichR11}
Alekhnovich, M.; and Razborov, A.~A. 2011.
\newblock {Satisfiability, Branch-Width and Tseitin tautologies}.
\newblock \emph{Comput. Complex.}, 20(4): 649--678.

\bibitem[{Ben-Sasson and Wigderson(2001)}]{BW01}
Ben-Sasson, E.; and Wigderson, A. 2001.
\newblock Short proofs are narrow - resolution made simple.
\newblock \emph{Journal of the ACM}, 48(2): 149--169.

\bibitem[{Bhatnagar and Peled(2015)}]{BhatnagarP15}
Bhatnagar, N.; and Peled, R. 2015.
\newblock Lengths of monotone subsequences in a Mallows permutation.
\newblock \emph{Probability Theory and Related Fields}, 161: 719--780.

\bibitem[{Biere et~al.(2020)Biere, Fazekas, Fleury, and
  Heisinger}]{BiereFazekasFleuryHeisinger2020}
Biere, A.; Fazekas, K.; Fleury, M.; and Heisinger, M. 2020.
\newblock {CaDiCaL}, {Kissat}, {Paracooba}, {Plingeling} and {Treengeling}
  Entering the {SAT Competition 2020}.
\newblock In Balyo, T.; Froleyks, N.; Heule, M.; Iser, M.; J{\"a}rvisalo, M.;
  and Suda, M., eds., \emph{Proc.~of {SAT Competition} 2020 -- Solver and
  Benchmark Descriptions}, volume B-2020-1 of \emph{Department of Computer
  Science Report Series B}, 51--53. University of Helsinki.

\bibitem[{Biere et~al.(2021)Biere, Heule, van Maaren, and
  Walsh}]{HandbookSAT2nd}
Biere, A.; Heule, M.; van Maaren, H.; and Walsh, T., eds. 2021.
\newblock \emph{Handbook of Satisfiability - Second Edition}, volume 336 of
  \emph{Frontiers in Artificial Intelligence and Applications}.
\newblock {IOS} Press.
\newblock ISBN 978-1-64368-160-3.

\bibitem[{Bodlaender(1998)}]{Bodlaender98}
Bodlaender, H.~L. 1998.
\newblock A Partial \emph{k}-Arboretum of Graphs with Bounded Treewidth.
\newblock \emph{Theor. Comput. Sci.}, 209(1-2): 1--45.

\bibitem[{Bonacina, Bonet, and Levy(2023)}]{BonacinaBL23}
Bonacina, I.; Bonet, M.~L.; and Levy, J. 2023.
\newblock Polynomial Calculus for MaxSAT.
\newblock In Mahajan, M.; and Slivovsky, F., eds., \emph{26th International
  Conference on Theory and Applications of Satisfiability Testing, {SAT} 2023,
  July 4-8, 2023, Alghero, Italy}, volume 271 of \emph{LIPIcs}, 5:1--5:17.
  Schloss Dagstuhl - Leibniz-Zentrum f{\"{u}}r Informatik.

\bibitem[{Bryant, Biere, and Heule(2022)}]{BBH22}
Bryant, R.~E.; Biere, A.; and Heule, M. J.~H. 2022.
\newblock Clausal Proofs for Pseudo-Boolean Reasoning.
\newblock In Fisman, D.; and Rosu, G., eds., \emph{Tools and Algorithms for the
  Construction and Analysis of Systems}, 443--461. Cham: Springer International
  Publishing.
\newblock ISBN 978-3-030-99524-9.

\bibitem[{Buss and Thapen(2019)}]{buss2019drat}
Buss, S.; and Thapen, N. 2019.
\newblock DRAT proofs, propagation redundancy, and extended resolution.
\newblock In \emph{International Conference on Theory and Applications of
  Satisfiability Testing}, 71--89. Springer.

\bibitem[{Chandran and Subramanian(2003)}]{ChandranS03}
Chandran, L.~S.; and Subramanian, C.~R. 2003.
\newblock A spectral lower bound for the treewidth of a graph and its
  consequences.
\newblock \emph{Inf. Process. Lett.}, 87(4): 195--200.

\bibitem[{Chew and Heule(2020)}]{ChewH20}
Chew, L.; and Heule, M. J.~H. 2020.
\newblock Sorting Parity Encodings by Reusing Variables.
\newblock In Pulina, L.; and Seidl, M., eds., \emph{Theory and Applications of
  Satisfiability Testing - {SAT} 2020 - 23rd International Conference, Alghero,
  Italy, July 3-10, 2020, Proceedings}, volume 12178 of \emph{Lecture Notes in
  Computer Science}, 1--10. Springer.

\bibitem[{de~Colnet and Mengel(2023)}]{deColnetM23}
de~Colnet, A.; and Mengel, S. 2023.
\newblock Characterizing Tseitin-Formulas with Short Regular Resolution
  Refutations.
\newblock \emph{J. Artif. Intell. Res.}, 76: 265--286.

\bibitem[{Enright et~al.(2021)Enright, Meeks, Pettersson, and
  Sylvester}]{EnrightMPS21}
Enright, J.~A.; Meeks, K.; Pettersson, W.; and Sylvester, J. 2021.
\newblock Tangled Paths: {A} Random Graph Model from Mallows Permutations.
\newblock \emph{CoRR}, abs/2108.04786.

\bibitem[{Fichte et~al.(2023)Fichte, Berre, Hecher, and
  Szeider}]{FichteLeberreHecherSzeider23}
Fichte, J.~K.; Berre, D.~L.; Hecher, M.; and Szeider, S. 2023.
\newblock The Silent (R)evolution of {SAT}.
\newblock \emph{Communications of the ACM}, 66(6): 64--72.

\bibitem[{Friedman(2003)}]{Friedman03}
Friedman, J. 2003.
\newblock A proof of Alon's second eigenvalue conjecture.
\newblock In Larmore, L.~L.; and Goemans, M.~X., eds., \emph{Proceedings of the
  35th Annual {ACM} Symposium on Theory of Computing, June 9-11, 2003, San
  Diego, CA, {USA}}, 720--724. {ACM}.

\bibitem[{Galesi, Talebanfard, and Tor{\'{a}}n(2020)}]{GalesiTT20}
Galesi, N.; Talebanfard, N.; and Tor{\'{a}}n, J. 2020.
\newblock {Cops-Robber Games and the Resolution of Tseitin Formulas}.
\newblock \emph{{ACM} Trans. Comput. Theory}, 12(2): 9:1--9:22.

\bibitem[{Gelder(2008)}]{Gelder08}
Gelder, A.~V. 2008.
\newblock Verifying {RUP} Proofs of Propositional Unsatisfiability.
\newblock In \emph{International Symposium on Artificial Intelligence and
  Mathematics, {ISAIM} 2008, Fort Lauderdale, Florida, USA, January 2-4, 2008}.

\bibitem[{Goldberg and Novikov(2003)}]{GoldbergN03}
Goldberg, E.~I.; and Novikov, Y. 2003.
\newblock Verification of Proofs of Unsatisfiability for {CNF} Formulas.
\newblock In \emph{2003 Design, Automation and Test in Europe Conference and
  Exposition {(DATE} 2003), 3-7 March 2003, Munich, Germany}, 10886--10891.
  {IEEE} Computer Society.

\bibitem[{Han and Jiang(2012)}]{HJ12}
Han, C.-S.; and Jiang, J.-H.~R. 2012.
\newblock When Boolean Satisfiability Meets Gaussian Elimination in a Simplex
  Way.
\newblock In Madhusudan, P.; and Seshia, S.~A., eds., \emph{Computer Aided
  Verification}, 410--426. Berlin, Heidelberg: Springer Berlin Heidelberg.

\bibitem[{Harvey and Wood(2018)}]{HarveyW18}
Harvey, D.~J.; and Wood, D.~R. 2018.
\newblock The treewidth of line graphs.
\newblock \emph{J. Comb. Theory, Ser. {B}}, 132: 157--179.

\bibitem[{Heule, Jr., and Wetzler(2013)}]{HeuleHW13}
Heule, M.; Jr., W. A.~H.; and Wetzler, N. 2013.
\newblock Trimming while checking clausal proofs.
\newblock In \emph{Formal Methods in Computer-Aided Design, {FMCAD} 2013,
  Portland, OR, USA, October 20-23, 2013}, 181--188. {IEEE}.

\bibitem[{Itsykson et~al.(2020)Itsykson, Knop, Romashchenko, and
  Sokolov}]{ItsyksonKRS20}
Itsykson, D.; Knop, A.; Romashchenko, A.~E.; and Sokolov, D. 2020.
\newblock On OBDD-based Algorithms and Proof Systems that Dynamically Change
  the order of Variables.
\newblock \emph{J. Symb. Log.}, 85(2): 632--670.

\bibitem[{Itsykson, Riazanov, and Smirnov(2022)}]{ItsyksonRS22}
Itsykson, D.; Riazanov, A.; and Smirnov, P. 2022.
\newblock Tight Bounds for Tseitin Formulas.
\newblock In Meel, K.~S.; and Strichman, O., eds., \emph{25th International
  Conference on Theory and Applications of Satisfiability Testing, {SAT} 2022,
  August 2-5, 2022, Haifa, Israel}, volume 236 of \emph{LIPIcs}, 6:1--6:21.
  Schloss Dagstuhl - Leibniz-Zentrum f{\"{u}}r Informatik.

\bibitem[{J{\"a}rvisalo, Heule, and Biere(2012)}]{JHB12}
J{\"a}rvisalo, M.; Heule, M. J.~H.; and Biere, A. 2012.
\newblock Inprocessing Rules.
\newblock In Gramlich, B.; Miller, D.; and Sattler, U., eds., \emph{Automated
  Reasoning}, 355--370. Berlin, Heidelberg: Springer Berlin Heidelberg.
\newblock ISBN 978-3-642-31365-3.

\bibitem[{Kim and Wormald(2001)}]{KimW01}
Kim, J.~H.; and Wormald, N.~C. 2001.
\newblock Random Matchings Which Induce Hamilton Cycles and Hamiltonian
  Decompositions of Random Regular Graphs.
\newblock \emph{J. Comb. Theory, Ser. {B}}, 81(1): 20--44.

\bibitem[{Mallows(1957)}]{mallows1957}
Mallows, C.~L. 1957.
\newblock Non-null ranking models. I.
\newblock \emph{Biometrika}, 44(1/2): 114--130.

\bibitem[{{Marques-Silva}, Lynce, and Malik(2009)}]{DBLP:series/faia/SilvaLM09}
{Marques-Silva}, J.~P.; Lynce, I.; and Malik, S. 2009.
\newblock Conflict-Driven Clause Learning {SAT} Solvers.
\newblock In \emph{Handbook of Satisfiability}. IOS Press.

\bibitem[{Philipp and Rebola-Pardo(2016)}]{PR16}
Philipp, T.; and Rebola-Pardo, A. 2016.
\newblock DRAT Proofs for XOR Reasoning.
\newblock In Michael, L.; and Kakas, A., eds., \emph{Logics in Artificial
  Intelligence}, 415--429. Cham: Springer International Publishing.

\bibitem[{Pipatsrisawat and Darwiche(2011)}]{CDCLRes}
Pipatsrisawat, K.; and Darwiche, A. 2011.
\newblock On the power of clause-learning SAT solvers as resolution engines.
\newblock \emph{Artificial Intelligence}, 175(2): 512 -- 525.

\bibitem[{Robinson(1963)}]{Rob63}
Robinson, J.~A. 1963.
\newblock Theorem-Proving on the Computer.
\newblock \emph{Journal of the ACM}, 10(2): 163--174.

\bibitem[{Roussel(2011)}]{Roussel11}
Roussel, O. 2011.
\newblock Controlling a solver execution with the runsolver tool.
\newblock \emph{Journal on Satisfiability, Boolean Modeling and Computation},
  7(4): 139--144.

\bibitem[{Sinz and Biere(2006)}]{BDDER}
Sinz, C.; and Biere, A. 2006.
\newblock Extended Resolution Proofs for Conjoining BDDs.
\newblock In Grigoriev, D.; Harrison, J.; and Hirsch, E.~A., eds.,
  \emph{Computer Science -- Theory and Applications}, 600--611. Berlin,
  Heidelberg: Springer Berlin Heidelberg.

\bibitem[{Soos(2012)}]{Soos12}
Soos, M. 2012.
\newblock Enhanced Gaussian Elimination in DPLL-based SAT Solvers.
\newblock In Berre, D.~L., ed., \emph{POS-10. Pragmatics of SAT}, volume~8 of
  \emph{EPiC Series in Computing}, 2--14. EasyChair.

\bibitem[{Soos, Nohl, and Castelluccia(2009)}]{SNC09}
Soos, M.; Nohl, K.; and Castelluccia, C. 2009.
\newblock Extending SAT Solvers to Cryptographic Problems.
\newblock In Kullmann, O., ed., \emph{Theory and Applications of Satisfiability
  Testing - SAT 2009}, 244--257. Berlin, Heidelberg: Springer Berlin
  Heidelberg.
\newblock ISBN 978-3-642-02777-2.

\bibitem[{Tamaki(2022)}]{Tamaki22}
Tamaki, H. 2022.
\newblock Heuristic Computation of Exact Treewidth.
\newblock In Schulz, C.; and U{\c{c}}ar, B., eds., \emph{20th International
  Symposium on Experimental Algorithms, {SEA} 2022, July 25-27, 2022,
  Heidelberg, Germany}, volume 233 of \emph{LIPIcs}, 17:1--17:16. Schloss
  Dagstuhl - Leibniz-Zentrum f{\"{u}}r Informatik.

\bibitem[{Tseitin(1968)}]{Tseitin68}
Tseitin, G. 1968.
\newblock {On the Complexity of Derivation in Propositional Calculus}.
\newblock \emph{Studies in Constructive Mathematics and Mathematical Logic},
  Part 2: 115--125.

\bibitem[{Tseitin({1983})}]{Tseitin83}
Tseitin, G.~S. {1983}.
\newblock \emph{{On the Complexity of Derivation in Propositional Calculus}},
  466--483.
\newblock {Springer Berlin Heidelberg}.
\newblock ISBN 978-3-642-81955-1.

\bibitem[{Urquhart(1987)}]{Urquhart87}
Urquhart, A. 1987.
\newblock {Hard Examples for Resolution}.
\newblock \emph{J. {ACM}}, 34(1): 209--219.

\end{thebibliography}

\end{document}